\documentclass[11pt]{article}
\usepackage[margin=1in]{geometry}
\usepackage{amsmath,amssymb,amsthm,bm,mathtools}
\usepackage{natbib}
\usepackage{booktabs}
\usepackage{enumitem}
\usepackage{hyperref}
\usepackage{dsfont}
\hypersetup{colorlinks=true,citecolor=blue,linkcolor=blue,urlcolor=blue}
\usepackage{outlines}
\usepackage[justification=centering]{caption}
\usepackage{subcaption} 
\usepackage{bm}
\usepackage{bbm}
\usepackage{textcomp}
\usepackage{xcolor}
\usepackage[nocomma]{optidef}
\usepackage{tabularx}
\usepackage{comment}
\usepackage{tikz}

\newcommand{\bs}[1]{\boldsymbol{#1}}
\newcommand{\mc}[1]{\mathcal{#1}}

\newcommand{\Prb}{\mathbb{P}}
\newcommand{\E}{\mathbb{E}}
\newcommand{\FDR}{\mathrm{FDR}}
\newcommand{\FDP}{\mathrm{FDP}}
\newcommand{\BH}{\mathrm{BH}}

\newcommand{\cI}{\mathcal{I}}
\newcommand{\barPsi}{\overline{\Psi}}

\newcommand{\ignore}[1]{}

\theoremstyle{plain}
\newtheorem{theorem}{Theorem}[section]
\newtheorem{lemma}[theorem]{Lemma}

\newtheorem{result}[theorem]{Result}
\theoremstyle{definition}
\newtheorem{definition}[theorem]{Definition}
\theoremstyle{remark}
\newtheorem{remark}[theorem]{Remark}

\title{Further Results on Controlling the False Discovery Rate in Two-Sided Gaussian Mean Testing}
\author{
  Deepra Ghosh \qquad\qquad Sanat K. Sarkar \\[1ex]
  \normalsize Department of Statistics, Operations and Data Science \\
  \normalsize Temple University
}
\date{August 2026}

\begin{document}
\maketitle
\begin{abstract}
The recent work of \citet{GhoshSarkar2025} introduced Positive Tail Dependence Under the Null (PTDN) and developed Generalized Shifted Benjamini--Hochberg (BH) procedures for two-sided Gaussian $z$- and $t$-testing under known covariance structures. This paper develops further consequences of that framework. First, we derive explicit dependence-adaptive lower and upper bounds for the FDR of the original BH procedure in terms of the conditional variance parameters $\tau_i=1-R_i^2$, where $R_i^2$ is the squared multiple correlation between the $i$th statistic and the remaining coordinates. These bounds recover the exact BH FDR under independence and provide finite-sample, covariance-specific information complementary to generic bounds. We also identify conditions under which the coordinate-specific calibration of shifted BH can provide a rejection advantage over the original BH procedure. Second, we consider the practically important setting in which the covariance matrix is unknown but an independent Wishart estimator is available. Using simultaneous lower confidence bounds for the $\tau_i$'s, we construct a confidence-bound shifted BH procedure and establish finite-sample FDR control. To our knowledge, this is the first shifted-BH-type procedure with a finite-sample guarantee for two-sided Gaussian mean testing under a completely unknown covariance matrix estimated independently. Numerical studies illustrate the behavior of the covariance-adaptive bounds, the potential advantage of shifted BH over BH, and the performance of confidence-bound shifting under unknown covariance.
\end{abstract}


\section{Introduction}
\label{sec:introduction}
The multiplicity correction introduced by \citet{BenjaminiHochberg1995}, that is targeted to control the false discovery rate (FDR), is considered the foundational work that nurtured decades of insightful research in this domain. The method, commonly known as the Benjamini-Hochberg (BH) procedure,  is simple, transparent and easily implemented across circumstances in the real world. Despite being widely implemented as an acceptable form of correction, the theoretical guarantees of its control are scarce. Going beyond the control under independence in the same seminal work, the sole structure among the associated test statistics that guarantees finite-sample control involves a positive regression dependency, \citet{BenjaminiYekutieli2001, Sarkar2002}. Asymptotic control was also established under weak local dependence structures in \citet{StoreyEtAl2004}. Improvements have been made over the BH procedure by introducing adaptive measures, \citet{JinCai2007, LiangNettleton2012}, weighted measures, \citet{Genovese2006, Ignatiadis2016} and multi-stage alternatives \citet{BenjaminiKriegerYekutieli2006, BlanchardRoquain2009}. Several numerical studies have showed reliable FDR control of the BH method, \citet{Farcomeni2006, ReinerBenaim2007}, that has elicited its widespread use across scientific domains notwithstanding its lack of theoretical validity.

The recent work of \citet{GhoshSarkar2025} introduced \emph{Positive Left-Tail Dependence Under the Null} (PLTDN) as a dependence framework for developing finite-sample FDR-controlling BH-type step-up procedures for testing Gaussian means against two-sided alternatives under dependence. Building on the leave-one-out representation of the BH FDR, that work established a general FDR theorem for PLTDN-satisfying evidence measures, referred to as shifted $p$-values, and developed a broad class of Generalized Shifted BH (GSBH) procedures encompassing and extending the shifted-BH methodology of \citet{SarkarZhang2025} for known covariance structures. The resulting framework substantially broadened the scope of finite-sample FDR control for two-sided Gaussian $z$- and $t$-testing, including regression-based variable selection.

The present paper continues this line of research. We begin by revisiting the development of PLTDN, adopting the simpler nomenclature \emph{Positive Tail Dependence Under the Null} (PTDN), and provide further insight into this notion of positive dependence, particularly its relationship with positive regression dependence on a subset (PRDS). We then develop three further consequences of the PTDN framework that were not explored in \citet{GhoshSarkar2025}.

First, we investigate the behavior of the original BH procedure itself. For two-sided $z$-tests, we derive explicit finite-sample lower and upper bounds on its FDR in terms of the conditional variance parameters $\tau_i=1-R_i^2$, where $R_i^2$ denotes the squared multiple correlation between $X_i$ and the remaining variables. We also derive a corresponding upper bound for two-sided $t$-tests. Building on the GSBH methodology developed under PTDN in \citet{GhoshSarkar2025}, these results provide a dependence-adaptive characterization of BH under Gaussian dependence. The bounds recover the exact BH FDR under independence, specialize naturally to structured covariance models such as equi-correlated Gaussian distributions, and complement existing generic bounds \citep{BenjaminiYekutieli2001,Su2018}. Recent work \citep{Dobriban2026,Lei2026} has shown that BH admits neither universal finite-sample FDR control nor a universal multiplicative FDR bound under arbitrary Gaussian dependence. Our results address a complementary question by quantifying the finite-sample behavior of BH for a given covariance structure rather than its worst-case behavior over unrestricted dependence.

Second, we investigate the power of shifted BH procedures. Although the shifted BH procedures of \citet{SarkarZhang2025} provide exact finite-sample FDR control, their potential power advantages over the original BH procedure have not been systematically explored. Focusing on shifted BH method 1 (SBH-1) of \citet{SarkarZhang2025}, we identify sufficient conditions under which its coordinate-specific calibration yields a strict power advantage over BH and provide a simple illustrative example demonstrating this phenomenon. The example highlights how the conditional variance parameters can enlarge the effective rejection thresholds for highly dependent coordinates, allowing SBH-1 to detect signals missed by BH while retaining finite-sample FDR control.

Our third contribution concerns unknown covariance structures. The methodology of \citet{GhoshSarkar2025} assumes that the covariance matrix is known. In many applications, however, dependence information is available only through an independent estimate obtained from historical studies, pilot experiments, replicate measurements, or other independent sources of covariance information. To address this practically important setting, we consider $\bs{X} \sim N_d(\bs{\mu},\bs{\Sigma})$ and $\bs{A} \sim W_d (n,\bs{\Sigma})$ where $\bs{X} \perp \bs{A}$. 
We develop, to the best of our knowledge, the first shifted-BH-type procedure possessing exact finite-sample FDR control in this setting. The construction uses simultaneous lower confidence bounds for the conditional variance parameters underlying the PTDN argument, leading to a simple confidence-bound shifted BH procedure that approaches its oracle counterpart as the Wishart degrees of freedom increase.

These developments share a common theme. While \citet{GhoshSarkar2025} established PTDN as a dependence principle underlying the construction of finite-sample FDR-controlling procedures, the present paper shows that the same framework also provides a dependence-adaptive understanding of the original BH procedure, clarifies when and why shifted BH methods can gain power, and extends the methodology to unknown covariance structures. Together, the two papers provide a broader framework for finite-sample FDR control in two-sided Gaussian mean testing under dependence.

The remainder of the paper is organized as follows. Section~\ref{sec:review} revisits the PTDN framework and recalls, without proof, results from \citet{SarkarZhang2025} and \citet{GhoshSarkar2025} needed in the sequel, including SBH-1 and the conditions underlying its potential power advantage over BH. Section~\ref{sec:bounds} develops dependence-adaptive lower and upper bounds for the FDR of the original BH procedure and discusses their relationship with existing generic and asymptotic bounds. Section~\ref{sec:unknown} develops confidence-bound shifted BH procedures for unknown covariance matrices estimated from an independent Wishart sample. Section~\ref{sec:numerical} presents numerical studies illustrating the power advantage of SBH-1 over BH for two-sided $z$-tests, comparing the covariance-adaptive upper bound with the generic Benjamini--Yekutieli bound under weak to moderate dependence, and assessing the performance of the proposed confidence-bound shifted BH procedure when the covariance matrix is unknown. Section~\ref{sec:discussion} concludes with a discussion and directions for future research.

\section{Inherited Tail Dependence Machinery}\label{sec:review}
This section briefly reviews the work of \citet{GhoshSarkar2025}, emphasizing those aspects needed for the subsequent developments. Specifically, we revisit the notion of PLTDN introduced there, henceforth referred to more simply as PTDN, provide some additional observations on this dependence condition, and recall several basic PTDN-related results for testing Gaussian means against two-sided alternatives using two-sided $z$- and $t$-tests.

\subsection{Introducing PTDN} 
\label{subsec:PTDN intro} 
Let $\bs P=(P_1,\ldots,P_d)$, where each $P_i$ is an evidence measure with smaller values indicating stronger evidence against the null hypotheses $H_i$ for  $i=1,\ldots,d$, and let $\cI_0$ denote the set of true-null indices. Consider a step-up testing procedure with critical constants $\alpha_1< \cdots < \alpha_d$, defined as follows: Order the $P_i$'s as
$P_{(1)}\leq\cdots\leq P_{(d)}$ and define $R(\bs P)=\max\{r:P_{(r)}\leq\alpha_r\}$,
if the maximum exists, and $R(\bs P)= 0$ otherwise. Reject all $H_i$ such that $P_i\leq P_{(R)}$, with no rejections when $R=0$.

Let $R(\bs P_{-i})$ be the leave-one-out rejection count obtained from $\bs{P}_{-i}$ with the critical constants $\alpha_2,\ldots,\alpha_d$. Then, a standard leave-one-out identity for the FDR of this step-up procedure is
\begin{eqnarray}\label{eqn:FDR}
\FDR & = & \sum_{i\in \cI_0}\E\left[
\frac{\mathbb{I} \{P_i\le \alpha_{R(\bs P_{-i})+1}\}}{R(\bs P_{-i})+1}\right ] \nonumber \\ & = & \sum_{i\in \cI_0} \Prb(P_i \le \alpha_1) + \sum_{i\in\cI_0}\sum_{r=1}^{d-1} \left[
\frac{\Prb \{P_i\le \alpha_{r+1}\mid R(\bs{P}_{-i}) \ge r\}}{r+1} - \right. \nonumber \\ & & \qquad \left. \frac{\Prb\{P_i\le \alpha_{r}\mid R(\bs{P}_{-i}) \ge r \}}{r}
\right] \Prb(R(\bs{P}_{-i}) \ge r) \end{eqnarray}
This representation, originated, for instance, in \citet{Sarkar2002,Sarkar2008}, highlights that the dependence of $P_i$ on $\bs P_{-i}$, rather than the reverse dependence emphasized in PRDS, lies at the heart of finite-sample FDR control, and has motivated \citet{GhoshSarkar2025} to introduce the following notion of \emph{Positive Tail Dependence Under the Null} (PTDN), referred to there as \emph{Positive Left-Tail Dependence Under the Null} (PLTDN).

\begin{definition}\label{Definition1}
For a fixed $c \in (0,1)$, the $P_i$'s are said to satisfy \emph{Positive Tail Dependence Under the Null} (PTDN) over $(0,c)$ if, for every $i\in \cI_0$, every increasing function $g(\cdot)$ of $\bs P_{-i}$, and every fixed $t>0$,
\[
\frac{\Prb(P_i\le u\mid g(\bs P_{-i})\le t)}{u}
\]
is decreasing in $u\in(0,c)$.
When $c=1$, the $P_i$'s are simply said to satisfy PTDN.
\end{definition}

\begin{remark}\label{Remark1} (i) The condition $\Prb\left(P_i \le u \mid \bs{P}_{-i}\right)/u \;
\downarrow \; u \in (0,c),$ being  stronger than PTDN, may itself  be viewed as a stronger form of PTDN. In particular, it holds whenever the conditional distribution function of $P_i$ given $\bs{P}_{-i}$ is concave in $u$, for each $i\in \cI_0$.

(ii) PTDN is invariant under co-monotone transformations. In particular, suppose
$P_i=\phi_i(T_i)$, where $\phi_i(\cdot)$ is monotone. If $\phi_i(\cdot)$ is decreasing, then the PTDN condition for the $P_i$'s is equivalent to $\Prb\left(T_i \ge x \mid g(\bs{T}_{-i}) \ge s \right)/ \Pr(T_i \ge x)$ being increasing in $x$, for every $i\in \cI_0$, every increasing function $g(\cdot)$, and every fixed constant $s>0$. Thus, PTDN translates into a form of positive dependence between $T_i$ and $\bs{T}_{-i}$ through their right tails. If $\phi_i(\cdot)$ is increasing, PTDN translates into an analogous positive dependence condition through the left tails of the $T_i$'s. Keeping in hindsight the natural left-tail dependence of the $p$-values but logical transformed right-tail dependence of the test statistics, we register the nomenclature of PTDN as the standard name for the dependency wherein the PLTDN in \citet{GhoshSarkar2025} referred only to the dependence among the $p$-values. The condition represents the same positive dependency of $P_i$ on the  conditioning event $\{h(\bs{P}_{-i})\geq t\}$ where $h(\cdot)$ is a decreasing function of $\bs{P}_{-i}$. 

(iii) The interpretation of PTDN as a form of positive dependence between $P_i$ and $\bs P_{-i}$, for each $i\in \cI_0$, follows from the fact that conditioning on stronger evidence against the remaining null hypotheses increases the tendency of $P_i$ to take smaller values relative to what would be expected under independence, with this relative increase becoming more pronounced deeper in the lower tail.
This provides a new perspective on positive dependence in multiple testing. Since the work of \citet{BenjaminiYekutieli2001}, \emph{Positive Regression Dependence on a Subset} (PRDS) has served as the principal framework for establishing FDR control under dependence. It is a global regression monotonicity condition imposed on the joint distribution of the test statistics, whereas PTDN is formulated directly in terms of the conditional lower-tail probabilities arising in the leave-one-out representation of the BH FDR. To clarify the relationship between the two notions, suppose the null $P_i$'s are marginally uniform under independence and let
\[
A_t=\{g(\bs P_{-i})\le t\},
\]
where $g(\cdot)$ is increasing. Then, for each null $P_i$,
\[
\frac{\Prb(P_i\le u\mid A_t)}{u}
=
\frac{\Prb(A_t\mid P_i\le u)}{\Prb(A_t)}
=
\frac{\mathbb{E}\{\mathbb{I}(A_t)\mid P_i\le u\}}{\Prb(A_t)}.
\]
Since $\mathbb{I}(A_t)$ is decreasing in $\bs P_{-i}$, the PRDS condition implies that the above ratio is decreasing in $u$, establishing that
\[
\mathrm{PRDS}\Longrightarrow\mathrm{PTDN}.
\]
Thus, PTDN is a weaker dependence condition, requiring the regression monotonicity property only for the lower-tail events arising in the leave-one-out FDR identity rather than for all decreasing functions. In this sense, PTDN may be viewed as the dependence condition naturally associated with BH-type procedures.

(iv) Two related notions of PTDN are used. When we say that $(P_1,\ldots,P_d)$ satisfies PTDN, we mean that, for every null hypothesis $H_i$, the corresponding evidence measure $P_i$ is PTDN on $\bs P_{-i}$ in the sense of Definition~\ref{Definition1}. Thus, PTDN for a collection is simply the coordinatewise PTDN property holding simultaneously for all null coordinates. However, we often work with transformed evidence measures $\widehat P_i$, constructed from the original $P_i$'s. In such cases, the statement that ``$\widehat P_i$ is PTDN on $\bs P_{-i}$'' refers only to the coordinatewise property for the $i$th transformed evidence measure and does not imply that the entire collection $(\widehat P_1,\ldots,\widehat P_d)$ satisfies PTDN. This distinction is important because the following general FDR theorem requires only the coordinatewise PTDN property for each null evidence measure, allowing some coordinates to be replaced by suitable PTDN-preserving transformations whenever the original evidence measures themselves fail to satisfy PTDN.

\end{remark}
The following theorem forms the theoretical foundation for some of the developments in this paper.

\begin{theorem}\label{Theorem1}(\cite{GhoshSarkar2025})
Consider a step-up test based on ${P}_i$'s and critical constants ${\alpha_i} = i \widetilde{\alpha}/d$, $i=1,\ldots, d$, for some fixed $\widetilde{\alpha} \in (0,1)$. Let, for each $i \in \cI_0$, either ${P}_i$ is PTDN on ${\bs{P}}_{-i}$ over $(0,c)$, or if not, there is a stochastically smaller and increasing function of ${P}_i$, conditionally given ${\bs{P}}_{-i}$, say $\hat{P}_i$, which is PTDN on ${\bs{P}}_{-i}$ over $(0,c)$, for some $c \in (\widetilde{\alpha}, 1]$.  Then, the FDR of this step-up test is bounded above by $\sum_{i \in \cI_0 \bigcap \mc{S}^{c}} \textrm{Pr} ({P}_i \le \widetilde{\alpha}/d) + \sum_{i \in \cI_0 \bigcap \mc{S}} \textrm{Pr} (\hat{P}_i \le \widetilde{\alpha}/d)$, where $\mc{S} = \{i: \hat{P}_i \overset{\text{st}}{\preceq} {P}_i\}$.
\end{theorem}

\subsection{Basic PTDN-related results}
\label{subsec:PTDN results}
Let $\bs X\sim N_d(\bs\mu,\bs\Sigma)$, where $\bs\Sigma$ is a known correlation matrix. For two-sided $z$-tests, the original $p$-values for testing
\[
H_i:\mu_i=0
\quad\text{against}\quad
K_i:\mu_i\neq0,
\qquad i=1,\ldots,d,
\]
are $P_i=\bar{\Psi}_1(X_i^2)$, where $\bar{\Psi}_1(\cdot)$ denotes the survival function of a $\chi_1^2$ distribution. For two-sided $t$-tests, suppose
\[
\bs X\sim N_d(\bs\mu,\sigma^2\bs\Sigma),
\qquad
V\sim\sigma^2\chi_\nu^2,
\qquad
V\perp\bs X.
\]
Then the original $p$-values for the same testing problem are given by $P_i=\bar{\Psi}_{1,\nu}(X_i^2/V)$, where $\bar{\Psi}_{1,\nu}(\cdot)$ denotes the survival function of the distribution of $F_{1,\nu}/\nu$, equivalently the beta type-II distribution with shape parameters $(1,\nu)$.

For $0<\tau\le 1$, define
\begin{equation*}\label{eq:Gdef}
G_\tau(u)=\barPsi_1\!\left\{\tau\barPsi_1^{-1}(u)\right\}\; \mbox{and} \;  H_{\tau}(u)=\barPsi_{1,\nu}\!\left\{\tau\barPsi_{1,\nu}^{-1}(u)\right\}, \;
\qquad 0<u<1.
\end{equation*}
For notational convenience and to provide a unified treatment of the two testing settings, we use $F_{\tau}(\cdot)$ generically to denote $G_{\tau}(\cdot)$ for two-sided $z$-tests and $H_{\tau}(\cdot)$ for two-sided $t$-tests. Also, let $\tau_i := 1/((\bs{\Sigma}^{-1}))_{ii} = 1-R_i^2$, where $R_i^2$ is the squared multiple correlation between $X_i$ and $\bs{X}_{-i}$.

The following lemma, when combined with Theorem~\ref{Theorem1}, yields the finite-sample upper bounds for the FDR of step-up procedures based on $\widetilde P_{i,\tau}=F_{\tau}^{-1}(P_i)$, where the original $p$-values $P_i$ are defined as above. These transformed quantities are referred to as \emph{shifted $p$-values} \citep{SarkarZhang2025}.

\begin{lemma}\label{Lemma1} (\cite{GhoshSarkar2025})
For each $i \in \cI_0$, the following hold: \begin{enumerate}
    \item [(i)] If $\tau \leq \tau_i$, then $\widetilde{P}_{i,\tau} = F_{\tau}^{-1}(P_i)$ is PTDN on $\bs{P}_{-i}$.
    \item [(ii)] If $\tau > \tau_i$, then $\hat{P}_{i,\tau} = \widetilde{\alpha} \widetilde{P}_{i,\tau_i}/F_{\tau_i/\tau}^{-1}(\widetilde{\alpha})$ is increasing in $\widetilde{P}_{i,\tau}$, stochastically smaller than $\widetilde{P}_{i,\tau}$, and is PTDN on $\bs{P}_{-i}$ over $(0, \widetilde{\alpha}]$.
    \end{enumerate}
    \end{lemma}

To facilitate interpretation of the aforementioned upper bounds, we outline a proof of Lemma \ref{Lemma1} below, drawing on the necessary supporting results from \cite{GhoshSarkar2025}, which we state without proof.

First, recall that, for each $i\in\cI_0$,
\begin{equation}\label{eqn:Non-CentChi}
X_i^2\mid\bs X_{-i}
\ \overset{d}{=}\
\tau_i\chi_1^2\{\lambda_i(\bs X_{-i})\}, \; \mbox{where} \; \lambda_i(\bs X_{-i})= \frac{1}{\tau_i}\left\{\bs\Sigma_{-i,i}^{\prime}
\bs\Sigma_{-i,-i}^{-1}(\bs X_{-i}-\bs\mu_{-i})\right\}^{2} \end{equation}
is the conditional non-centrality parameter. Consequently, the non-central chi-square distribution with $1$ degree of freedom arises naturally in the FDR calculation involving two-sided $z$-tests, and for the corresponding two-sided $t$-tests, the relevant distribution is the corresponding non-central beta type-II distribution $B_{1,\nu}^{\prime}(\lambda)$, where $\nu$ denotes the degrees of freedom of the central chi-square variable in the denominator of the corresponding random variable imagined as a ratio of two independent chi-square variables..

Let $\bar{\Psi}_{1;\lambda}(\cdot)$ and $\bar{\Psi}_{1, \nu;\lambda}(\cdot)$ be the noncentral counterparts of $\bar{\Psi}_{1}(\cdot)$ and $\bar{\Psi}_{1, \nu}(\cdot)$, respectively, and, for any fixed $\tau \in (0,1)$,
\begin{equation*}
G_{\tau, \lambda}(u)=\bar{\Psi}_{1;\lambda}\!\left\{\tau\bar{\Psi}_1^{-1}(u)\right\}\; \mbox{and} \;  H_{\tau, \lambda}(u)=\bar{\Psi}_{1,\nu;\lambda}\!\left\{\tau\bar{\Psi}_{1,\nu}^{-1}(u)\right\}, \;
\qquad 0<u<1.
\end{equation*}
We let $F_{\tau, \lambda}(\cdot)$ generically to denote $G_{\tau,\lambda}(\cdot)$ for two-sided $z$-tests and $H_{\tau,\lambda}(\cdot)$ for two-sided $t$-tests.

\begin{result}\label{Result1}
For any $\tau\in(0,1]$ and $\lambda\geq0$, $F_{\tau,\lambda}(u)$ is concave in $u\in(0,1)$.
\end{result}

\begin{remark}\label{Remark2} When $\lambda=0$, $F_{\tau,\lambda}(u)$ reduces to $F_{\tau}(u)$. Result~\ref{Result1} therefore implies that, for any $\tau\in(0,1]$:
(i) $F_{\tau}(u)$ is concave, and hence its inverse $F_{\tau}^{-1}(u)=F_{1/\tau}(u)$ is convex in $u\in(0,1)$; (ii) since $F_{\tau}(0)=F_{\tau}^{-1}(0)=0$, the ratios $F_{\tau}(u)/u$ and $F_{\tau}^{-1}(u)/u$ are, respectively, decreasing and increasing in $u\in(0,1)$.

It should be noted that, for two-sided $z$-tests, the PTDN property of $\widetilde P_{i,\tau}$, or of its minorant $\widehat P_{i,\tau}$, stated in Lemma~\ref{Lemma1}, holds in the stronger form described in Remark~\ref{Remark1}(i), as a consequence of Result~\ref{Result1} and the above concavity and convexity properties. For two-sided $t$-tests, however, the PTDN condition holds in its original form, as given in Definition~\ref{Definition1}, and for that the following additional result is required.
\end{remark}

\begin{result}\label{Result2}
Let $V\sim\chi_\nu^2$. Then, for any fixed $\lambda\geq0$, $\theta\in(0,1]$, and any decreasing function $g(\cdot)$ of $V$,
\[
\frac{1}{u}
\E\left[
\bar{\Psi}_{m;\lambda}
\left\{
\theta V\bar{\Psi}_{m,\nu}^{-1}(u)
\right\}
g(V)
\right]
\]
is decreasing in $u\in(0,1)$.
\end{result}

The aforementioned upper bound obtained by applying Lemma~\ref{Lemma1} to Theorem~\ref{Theorem1} is given in the following theorem.

\begin{theorem}\label{Theorem2} (\cite{GhoshSarkar2025})
Consider the step-up procedure based on the $\widetilde P_{i,\tau}$'s, for a fixed $\tau\in(0,1]$, with critical constants
$\alpha_i= i\widetilde\alpha/{d}$, $i=1,\ldots,d$. Then,
\begin{eqnarray}
\FDR_{\mathrm{GSBH}}
&\leq&
\sum_{i\in\cI_0:\,\tau_i<\tau}
F_{\tau_i}
\left\{
\frac{1}{d}
F_{\tau_i/\tau}^{-1}(\widetilde\alpha)
\right\}
+
\sum_{i\in\cI_0:\,\tau_i \geq \tau}
F_{\tau}\left(\frac{\widetilde\alpha}{d}\right).
\end{eqnarray}
\end{theorem}
For given $\tau_i$'s, calibrating $\widetilde{\alpha}$ in the step-up critical constants for the $\widetilde{P}_{i,\tau}$'s
using these bounds yields the level-$\alpha$ generalized shifted BH (GSBH) procedures of \citet{GhoshSarkar2025}.

\subsection{Revisit of SBH-1}
\label{subsec:SBH1}
The GSBH methodology allows the shift parameter $\tau$ to vary coordinatewise and taking $\tau=\tau_i$ for each $P_i$ yields $\widetilde P_{i,\tau_i} =  F_{\tau_i}^{-1}(P_i)$, $i=1,\ldots,d$. These coordinate-specific shifted $p$-values satisfy the PTDN condition (see Result~\ref{Result1} and Lemma~\ref{Lemma1}) and give rise to the FDR-controlling SBH-1 procedure of \citet{SarkarZhang2025}, recalled below for subsequent reference:

\begin{definition} (\emph{SBH-1, controlling FDR at level $\alpha$}). The BH procedure based on $\widetilde{P}_{i,\tau_i} = F_{\tau_i}^{-1}(P_i)$, $i=1, \ldots, d$, at level $\widetilde{\alpha} = d F^{-1}(\alpha/d)$, where $F(u) = \frac{1}{d}\sum_{i=1}^d F_{\tau_i}(u)$, $u \in (0,1)$, \end{definition}

Thus, SBH-1 can be viewed as the coordinatewise analog of the GSBH methodology. The following theorem provides a sufficient condition for potential power advantage of SBH-1 over the corresponding BH procedure. 

\begin{theorem}
\label{Theorem3}
Consider SBH-1 at level $\alpha$ for two-sided $z$-tests, and let
$R_{\mathrm{SBH-1}}$ and $R_{\mathrm{BH}}$ denote the numbers of rejections made by SBH-1 and the corresponding BH procedure, respectively. For each $i$ and $r\geq1$, define
\[ \mathcal A_{i,r} = \left\{ P_i\in \left( \frac{r\alpha}{d}, G_{\tau_i} \left\{ rG^{-1}\left(\frac{\alpha}{d}\right)
\right \} \right], \; R_{\mathrm{SBH-1}} = r, \; R_{\mathrm{BH}}\leq r \right\}. \] If $\Pr(\mathcal A_{i,r}) >0$, for some $i\in\cI_0^c$ and $r\geq1$, then, on $\mathcal A_{i,r}$, SBH-1 rejects the false null hypothesis $H_i$, whereas BH does not.
\end{theorem}

In Section~\ref{sec:numerical}, we provide a toy example illustrating the phenomenon described in the preceding theorem.  

\section{Dependence-adaptive bounds for the BH FDR}\label{sec:bounds}
Here, we derive lower and upper bounds for the finite-sample FDR of the original BH procedure for testing Gaussian means against two-sided alternatives using $z$-tests, together with an upper bound for the corresponding $t$-tests. These bounds follow directly from the results of the preceding section.

Setting $\tau=1$ and $\widetilde{\alpha} = \alpha$ in Theorem~\ref{Theorem2}, so that GSBH reduces to the level $\alpha$ BH based on the original $p$-values, we have
\[
\operatorname{FDR}_{\mathrm{BH}}
\leq
\sum_{i\in\cI_0}
F_{\tau_i}
\left\{
\frac{1}{d}
F_{\tau_i}^{-1}(\alpha)
\right\},
\]
with
$F_{\tau_i}(u) = \bar{\Psi}_{1} \left\{\tau_i\bar{\Psi}_{1}^{-1}(u)\right\}$ for two-sided $z$-tests, whereas
$F_{\tau_i}(u) = \bar{\Psi}_{1,\nu}\left\{\tau_i\bar{\Psi}_{1,\nu}^{-1}(u)\right\}$ for two-sided $t$-tests.

These are the aforementioned upper bound for FDR$_{\textrm{BH}}$. The lower bound for the two-sided $z$-tests can be obtained as follows.
From \eqref{eqn:FDR} and \eqref{eqn:Non-CentChi}, we have
\begin{eqnarray}\label{eq:BHFDR1}
\operatorname{FDR}_{\mathrm{BH}}
& = &
\sum_{i\in\cI_0}\sum_{r=0}^{d-1}
\frac{1}{r+1}
\E\left[
\Prb\left\{P_i\le\frac{(r+1)\alpha}{d}\mid\bs P_{-i}\right\}
\mathbb{I}\{R(\bs P_{-i})=r\}
\right] \end {eqnarray}
Convexity of $G_{\tau_i}^{-1}(u)$ (see, Remark \ref{Remark2}) and concavity of
\[\Prb\left( G_{\tau_i}^{-1}(P_i) \le u \mid \bs{X}_{-i} \right ) = G_{\tau_i, \lambda_{i}\left(\bs{X}_{-i} \right))}(u) \] (see, Result \ref{Remark2}) in $u \in (0,1)$ yield

\begin{eqnarray} \label{eq:BHFDR2} & & \frac{1}{r+1} \Prb \left(P_i \le \frac{r+1}{d}\alpha \mid\bs{P}_{-i}\right ) = \frac{1}{r+1} G_{\tau_i, \lambda_i \left ( \bs{P}_{-i} \right)} \left ( G_{\tau_i}^{-1}\left ( \frac{(r+1)\alpha}{d}\right ) \right ) \nonumber \\ & \ge & \frac{1}{r+1} G_{\tau_i, \lambda_i \left ( \bs{P}_{-i} \right)} \left ( (r+1)G_{\tau_i}^{-1}\left ( \frac{\alpha}{d}\right ) \right ) \ge \frac{1}{d} G_{\tau_i, \lambda_i\left (\bs{P}_{-i} \right)} \left ( dG_{\tau_i}^{-1}\left ( \frac{\alpha}{d}\right ) \right ).\end{eqnarray} which does not depend on $r$. Therefore, using \eqref{eq:BHFDR2} in \eqref{eq:BHFDR1} and summing over $r$ and then averaging over $\bs{P}_{-i}$, we have
\[\operatorname{FDR}_{\mathrm{BH}} \ge \frac{1}{d}\sum_{i\in\cI_0}
\Prb \left ( P_i\le G_{\tau_i}\left(d G_{\tau_i}^{-1} \left(\frac{\alpha}{d} \right ) \right ) \right )=
\frac{1}{d} \sum_{i\in\cI_0}G_{\tau_i} \left( d G_{\tau_i}^{-1}\left(\frac{\alpha}{d} \right ) \right ) \]

Thus, more formally we have the following theorem providing finite-sample, dependence-adaptive bounds for the FDR of the original BH procedure.

\begin{theorem}\label{Theorem4}
For the BH procedure at level $\alpha$, the following bounds hold:
\begin{align}
\frac{1}{d}\sum_{i\in\cI_0}
G_{\tau_i}
\left\{
dG_{\tau_i}^{-1}\left(\frac{\alpha}{d}\right)
\right\}
\leq \FDR_{\BH}
\leq
\sum_{i\in\cI_0}
G_{\tau_i}
\left\{
\frac{1}{d}G_{\tau_i}^{-1}(\alpha)
\right\},
\qquad &\text{for two-sided $z$-tests},
\label{eq:FDRbounds}\\
\FDR_{\BH}
\leq
\sum_{i\in\cI_0}
H_{\tau_i}
\left\{
\frac{1}{d}H_{\tau_i}^{-1}(\alpha)
\right\},
\qquad &\text{for two-sided $t$-tests}.
\label{eq:FDRbound-t}
\end{align}
Under independence, the lower and upper bounds for the two-sided $z$-tests, as well as the upper bound for the two-sided $t$-tests, all reduce to $|\cI_0|/{d}\alpha$, the exact FDR of the BH procedure.
\end{theorem}

\begin{remark}
The bounds in Theorem~\ref{Theorem4} explicitly incorporate the conditional variance parameters $\tau_i=1-R_i^2$ and therefore exploit information about the underlying Gaussian dependence that is absent from generic distribution-free bounds. In particular, they recover the exact BH FDR under independence, where $\tau_i=1$ for all $i$, and adapt continuously to the covariance structure. This contrasts with the Benjamini--Yekutieli (BY) bound
\[
\operatorname{FDR}_{\mathrm{BH}}
\leq
\frac{|\cI_0|}{d}\alpha
\sum_{r=1}^{d}\frac{1}{r},
\]
which is valid under arbitrary dependence and applies far beyond the Gaussian setting, but does not distinguish among dependence structures.

The equi-correlated Gaussian model provides a particularly transparent illustration. If
\[
\bs\Sigma_\rho
=
(1-\rho)\bs I_d+\rho\bs 1_d\bs 1_d^\prime,
\]
then the conditional variance parameters are identical and given by
\[
\tau_i=\tau(\rho)
=
\frac{(1-\rho)\{1+(d-1)\rho\}}
     {1+(d-2)\rho},
\qquad i=1,\ldots,d.
\]
Thus, the bounds in Theorem~\ref{Theorem4} reduce to explicit finite-sample functions of the correlation parameter $\rho$. At $\rho=0$, $\tau(\rho)=1$ and the bounds recover the independence result, whereas increasing positive correlation decreases $\tau(\rho)$ and continuously modifies the resulting FDR bounds. Similar adaptation occurs for other structured covariance models, such as Toeplitz, autoregressive, and block-correlated structures, through their coordinate-specific $\tau_i$'s. Numerical studies in Section~\ref{sec:numerical} indicate that, under weak to moderate dependence, the resulting upper bound can be substantially sharper than the BY bound.

For two-sided Gaussian $z$-tests, Theorem~\ref{Theorem4} also provides a dependence-adaptive lower bound. This gives a finite-sample floor on the BH FDR for a specified covariance structure and, being bounded above by $(|\cI_0|/d)\alpha$, quantifies the extent to which BH can be conservative without suggesting FDR inflation above $\alpha$. The lower and upper bounds therefore provide a two-sided, covariance-adaptive characterization of the finite-sample BH FDR.

These results address a different question from other dependence-sensitive bounds for BH. The FDR-linking theorem of \citet{Su2018} yields a generic $O\{\alpha\log(1/\alpha)\}$ upper bound under positive regression dependence within the nulls, while \citet{Lei2026} establishes the sharper optimal order $O\{\alpha\sqrt{\log(1/\alpha)}\}$ for common-factor Gaussian models. Those results characterize generic or worst-case behavior over classes of dependence structures. In contrast, Theorem~\ref{Theorem4} conditions its characterization on the specified covariance matrix through the $\tau_i$'s. The equi-correlated model makes this distinction especially clear: rather than taking a worst case over a common-factor class, the proposed bound varies explicitly with the given value of $\rho$. It is therefore of interest to determine whether, under suitable common-factor asymptotics, this covariance-specific bound recovers the optimal order identified by \citet{Lei2026}, while potentially providing sharper finite-sample information over practically relevant correlation regimes.

The corresponding result for two-sided $t$-tests extends this covariance-adaptive perspective to the setting in which the common scale parameter is unknown and estimated independently. Thus, beyond the universally valid BY bound, the present results provide covariance-specific finite-sample information for both two-sided Gaussian $z$- and $t$-testing, while the lower bound available in the $z$-testing case further quantifies the possible conservativeness of the original BH procedure.
\end{remark}

\section{Unknown covariance matrix}\label{sec:unknown}
We now extend the shifted-BH methodology to the following practically important setting in which the covariance matrix is unknown but can be estimated independently:
\begin{equation}\label{eq:wishart-model}
\bs X\sim N_d(\bs\mu,\bs\Sigma),
\qquad
\bs A\sim W_d(n,\bs\Sigma),
\qquad
\bs X\perp\bs A,
\end{equation}
where $n \geq d$. As before, we consider multiple testing of $H_i:\mu_i=0$ against $K_i:\mu_i\neq 0$, $i=1,\ldots,d$, and seek a finite-sample FDR controlling procedure.

Our starting point is SBH-1 for two-sided $z$-tests, recalled in Section~\ref{subsec:SBH1}. We seek to adapt its dependence-adjustment mechanism to the unknown-covariance setting using the independently available Wishart matrix $\bs{A}$. Recall that SBH-1 transforms the original marginal statistics $X_i^2/\sigma_{ii}$, $i=1,\ldots, d$, into the conditional-variance-adjusted statistics $X_i^2/\sigma_{ii}(1-R_i^2)$, $i=1,\ldots, d$. In the present setting, the conditional variance $\sigma_{ii}(1-R_i^2) = 1/(\bs\Sigma^{-1})_{ii}$, is unknown. A natural initial approach is therefore to replace it by the corresponding sample quantity obtained from $\bs{A}$, $1/(\bs A^{-1})_{ii} = a_{ii}(1-\widehat R_i^2)$, $i=1,\ldots,d$, where $\widehat R_i^2$ denotes the sample squared multiple correlation computed from $\bs{A}$. This leads to the transformed statistics $X_i^2/ a_{ii}(1-\widehat R_i^2)$, $i=1,\ldots,d$, which provide the natural plug-in analogs of the PTDN-inducing statistics underlying SBH-1.

Unfortunately, however, the PTDN property underlying the proof of FDR control cannot, in general, be established for these plug-in transformations. To overcome this difficulty, we replace the unknown dependence parameters $\tau_i = 1-R_i^2, \; i=1,\ldots,d$, by suitably constructed lower confidence bounds $L_i$ derived from $\bs{A}$. Since the testing problem is scale invariant, we may assume without loss of generality that $\bs\Sigma$ is a correlation matrix, so that the relevant dependence parameters are precisely the $\tau_i$'s. Correspondingly, the oracle transformations $G_{\tau_i}(\cdot) = \bar{\Psi}_1(\tau_i\bar{\Psi}_1^{-1}(\cdot))$ are replaced by their data-dependent counterparts $G_{L_i}(\cdot)$, leading to the confidence-bound shifted $p$-values $G_{L_i}^{-1}(P_i), \; i=1,\ldots,d$. These reduce to their oracle counterparts when $L_i=\tau_i$. The resulting procedure may therefore be viewed as a confidence-bound analog of the PTDN-based shifted procedures developed for the known-covariance setting.

We are now ready to introduce the proposed procedure with unknown covariance matrix. To this end, let $C_i = a_{ii}(1-\hat{R}_i^2)$. Fix a confidence error level $\beta \in (0, \alpha)$ and choose values $\beta_i \in (0,\beta)$, $i=1, \ldots, d$, satisfying  $\sum_{i=1}^d \beta_i \le \beta$. Define $L_i = C_i/\bar{\Psi}_{\nu}^{-1}(\beta_i)$, where $\nu = n-d+1$. Since $\bs{A} \sim W_d(n, \bs{\Sigma})$, $C_i \sim \tau_i \chi_{n-d+1}^2$, independently of $\bs{A}_{-i, -i}$; see, for example, \cite{Anderson2003book}. Consequently, $\Prb_{\bs{\Sigma}}(L_i \le \tau_i) = 1 - \beta_i$, and, therefore, by Bonferroni, $\Prb_{\bs{\Sigma}}(L_i \le \tau_i, \; \mbox{for all} \; i = 1, \ldots, d ) \ge 1 - \beta$. We denote this simultaneous-coverage event by $\mc{E}_L$.

\begin{definition}
\label{DefCBSBH}
(\emph{Confidence-Bound Shifted BH (CBSBH)}). For each $i=1,\ldots,d$, let  $P_i=\bar{\Psi}_1(X_i^2)$ be the marginal Gaussian $p$-values. Given the lower confidence bounds $L_i$ constructed above, define the confidence-bound shifted $p$-values $\widetilde{P}_{i,L_i} = G_{L_i}^{-1}(P_i), \; i = 1,\ldots,d$. Further, let $\bar{G}_{L}(u) = \frac{1}{d}\sum_{i=1}^d G_{L_i}(u), \; u \in (0,1)$, and use the step-up procedure based on the $\widetilde{P}_{i,L_i}$'s with  critical constants $i \bar{G}_{L}^{-1}((\alpha - \beta)/d), \; i = 1, \ldots, d$.\end{definition}

\begin{theorem}\label {Theorem5} The CBSBH controls the FDR at level $\alpha$.
\end{theorem}

\begin{proof} Conditional on $\bs{A}$ and on the event $\mathcal{E}_L$, the quantities $L_1, \ldots, L_d$ are fixed and satisfy $L_i\le \tau_i$ for all $i=1, \ldots, d$. Hence, the $\widetilde{P}_{i,L_i}$'s satisfy the PTDN condition; see, Result \ref{Result1}(i). Consequently, the corresponding SBH-1 procedure at level $d\bar{G}_{L}^{-1}((\alpha - \beta)/d)$, that is, the step-up procedure based on the  $\widetilde{P}_{i,L_i}$'s with  critical constants $i \bar{G}_{L}^{-1}((\alpha - \beta)/d), \; i = 1, \ldots, d$, controls the FDR at level $\alpha-\beta$ on $\mathcal{E}_L$. Therefore, using FDP $< 1$,
\begin{eqnarray}
\FDR = \E(\FDP\mathds 1_{\mathcal E}) + \E (\FDP\mathds 1_{\mathcal E^c}) \le \alpha - \beta + \Prb(\mathcal E^c) \le \alpha - \beta + \beta = \alpha,
\end{eqnarray} which proves the theorem.
\end{proof}

\section{Numerical investigations}\label{sec:numerical}

\subsection{SBH-1 can outperform BH}
\label{subsec:SBH1 simulation}
Here, we provide a simple numerical example showing SBH-1 can be more powerful than the original BH in the two-sided Gaussian mean testing problem under dependence.

Suppose $\bs{X} \sim N_{10}(\bs{\mu},\bs{\Sigma})$, where the correlation matrix $\Sigma$ is such that the first three coordinates form a strongly equi-correlated block with the common correlation $0.8$ and the remaining seven coordinates are independent of one another and of the first block. We test $\mu_i=0$ against $\mu_i\neq 0$ using the usual two-sided $z$-test $p$-values $P_i = \bar{\Psi}_1(X_i^2)$, $i=1,\dots,10$, at the target FDR level $\alpha=0.05$.

For this covariance matrix, $\tau_i = 0.2888889$ for $i=1,2,3$, and $\tau_i=1$ for $i=4, \ldots, 10,$. Let $ G(u) = \frac{1}{10}\sum_{i=1}^{10} G_{\tau_i}(u)$, where $G_{\tau}(u) = \bar\Psi\bigl(\tau\,\bar\Psi^{-1}(u)\bigr)$, $0 \le u \le 1$. Then $G^{-1}(\alpha/10) \approx
1.194959\times 10^{-6}$. Therefore, in the present example, the SBH-1 is equivalently the BH at level $q = 10G^{-1}(\alpha/10) \approx 1.194959\times 10^{-5}$ applied to $\widetilde{P}_i = G_{\tau_i}^{-1}(P_i) = \bar{\Psi}_1 \left( X_i^2/\tau_i\right )$, $i=1, \ldots, 10$. By contrast, the original BH is applied to $P_i$, $i=1, \ldots, 10$, at level $q=0.05$.


In SBH-1, the critical constants are $ru_{\alpha}$, $r=1, \ldots, 10$, where $u_{\alpha} = G^{-1}(0.005) \approx 1.194959\times 10^{-6}$. Therefore, at the first step-up stage ($r=1$), the original-scale SBH-1 critical values are
\[
G_{\tau_i}(u_\alpha)
=
\begin{cases}
9.046944\times 10^{-3}, & i=1,2,3,\\
1.194959\times 10^{-6}, & i=4,\dots,10.
\end{cases}
\]
By contrast, the BH critical value at $r=1$ is simply $\alpha/10 = 0.005= 5 \times10^{-3}$. Thus, for the strongly dependent coordinates $i \in \mc{J}=\{1,2,3\}$, SBH-1 is considerably less stringent than BH.

Now, consider a data-realization for which the ordered $p$-values satisfy $P_1=0.008, P_2=\cdots=P_{10}=0.50$, and suppose that the signal is in the first coordinate. Then BH makes no rejection, because its first critical value is $0.005$ and $0.008 > 0.005$. However, SBH-1 rejects $H_{01}$, because for the first coordinate $0.008 < 0.009046944 = G_{\tau_1}(u_\alpha)$. Equivalently, the transformed $p$-value is
$\widetilde{P}_1 = G_{\tau_1}^{-1}(0.008) \approx 8.04669\times 10^{-7}$, which is below the first SBH-1 transformed cutoff $u_\alpha
\approx 1.194959\times 10^{-6}$. So, on this same data set, SBH-1 detects the false null while BH does not.

To show that this is not merely a hand-picked realization, consider the mean vector $\bs{\mu} = (2.6,0,0,0,0,0,0,0,0,0)^\top$,
so that only the first hypothesis is false, and it is located in the strongly correlated block.

A Monte Carlo experiment with $50{,}000$ replications gives the following empirical performance:

\ignore{\begin{center}
\begin{tabular}{lccc}
\toprule
Procedure & Empirical power for $H_{01}$ & Average number of rejections & Empirical FDR \\
\midrule
BH    & 0.4228 & 0.4937 & 0.0434 \\
SBH-1 & 0.5840 & 0.6200 & 0.0215 \\
\bottomrule
\end{tabular}
\end{center}}

\begin{center}
\begin{tabular}{lccc}
\toprule
Procedure & Empirical power  & Average number of rejections & Empirical FDR \\
\midrule
BH    & 0.21018 & 0.04991 & 0.04680 \\
SBH-1 & 0.24587 & 0.05125 & 0.01264 \\
\bottomrule
\end{tabular}
\end{center}

Thus, in this configuration, SBH-1 has higher power than BH. The gain comes from the fact that the non-null is placed in a coordinate with small $\tau_i$, that is, a coordinate highly predictable from the others. In such coordinates the transformation $G_{\tau_i}^{-1}(\cdot)$ shrinks the $p$-values, thereby relaxing the effective rejection threshold relative to BH.

This configuration also enables to understand the form of dependence necessary among the null and non-null hypotheses so that improved power over the BH can be achieved by SBH-1. The overarching philosophy of strong pairwise correlations existing only among the true effects while being completely independent of the null effects invokes a study into scalable improvement along with the introduction of minute correlations among the nulls. We go one further step to solidify the advantages of SBH-1 over BH in a broader simulation design. 

Consider an extended design where $\bs{X} \sim N_d (\bs{\mu}, \bs{\Sigma})$, with arbitrarily chosen $10\%$ of the mean effects as non-zero. $\bs{\Sigma}$ has correlations divided into two equicorrelated blocks: the set of nulls face pairwise correlations of $0.1$ while the set of non-nulls inherit strong pairwise correlations of $0.9$. For this correlation matrix, $\tau_i \in ( 0.1109606, 0.1892857)$ for $i \in \cI_0^C$ and $\tau_i \in (0.9090807, 0.9332609)$ for $i \in \cI_0$. When Monte Carlo simulations are run at level $\alpha = 0.05$ with varying number of moderately sized tests, the empirically averaged power can be visualized as in Fig~\ref{fig:SBH1_over_BH}. While controlling the FDR at desired level, SBH1 shows scalable improvement in detecting true signals over BH and dBH in cases where there is a strong affinity among the non-null effects. Moderate sizes of samples are often encountered in clinical trials whereas clustered impact of small number of significant traits are abundantly observed in genomic studies, underpinning the relevance and high odds of encountering such scenarios. Reverberating to the sufficient condition in Theorem~\ref{Theorem3}, our examples explicitly make the required probability positive by strategically constructing correlations encompassing the true effects only.
\begin{figure}[]
    \centering
        \includegraphics[width=\linewidth, height=0.35\textheight,
            keepaspectratio]{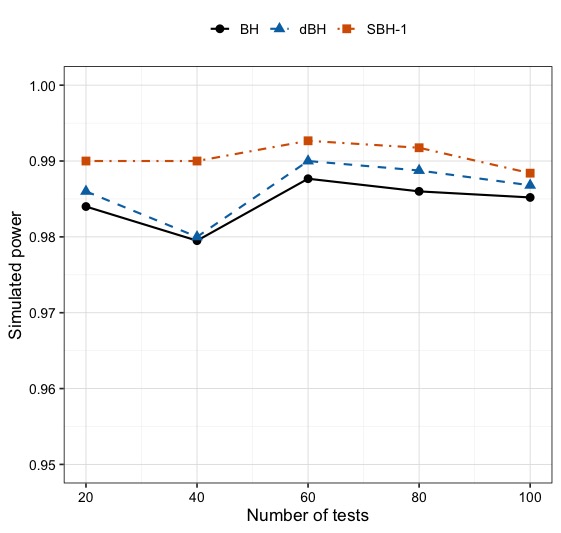}
    \caption{
        Power comparison of BH, dBH and SBH-1 methods
    }
    \label{fig:SBH1_over_BH}
\end{figure}

These examples illustrate the mechanism behind the power gain of SBH-1. When dependence is heterogeneous, coordinates with smaller $\tau_i=1-R_i^2$ receive more favorable calibration under SBH-1. If the true signals are concentrated in those coordinates, then SBH-1 can be strictly more powerful than BH, even though both procedures are applied to the same original set of two-sided marginal $p$-values. In particular, the present example shows that strong local dependence can be exploited by SBH-1 to recover signals that BH misses.

\subsection{Comparison of covariance-adaptive and BY bounds}
\label{supp-examples of bounds}


We explore scenarios where the covariance-adaptive upper bound would be sharper than the BY bound. To demonstrate the role of correlations being as important as the number of tests, we observed both the upper bounds across varying number of tests. Three individual correlation structures were considered, where one is a Toeplitz matrix, one compound symmetric and one has a block structure where elements within a group are equicorrelated. Small deviations from independence equates to small correlations incorporated in the three structures.

\begin{figure}[]
    \centering

    \begin{subfigure}[t]{0.5\linewidth}
        \centering
        \includegraphics[width=\linewidth, height=0.27\textheight,
            keepaspectratio]{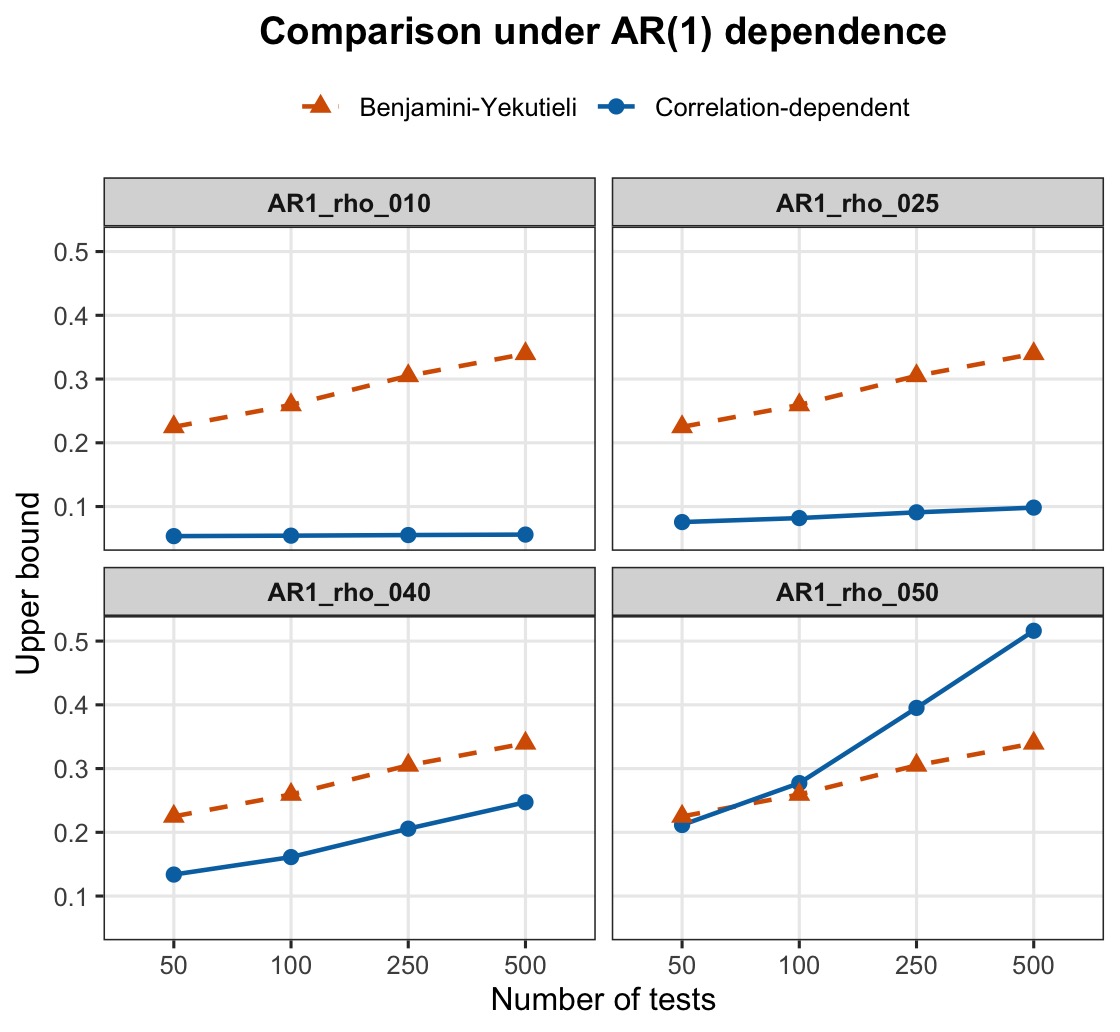}
        \caption{}
    \end{subfigure}
    \begin{subfigure}[t]{0.5\linewidth}
        \centering
        \includegraphics[width=\linewidth, height=0.27\textheight,
            keepaspectratio]{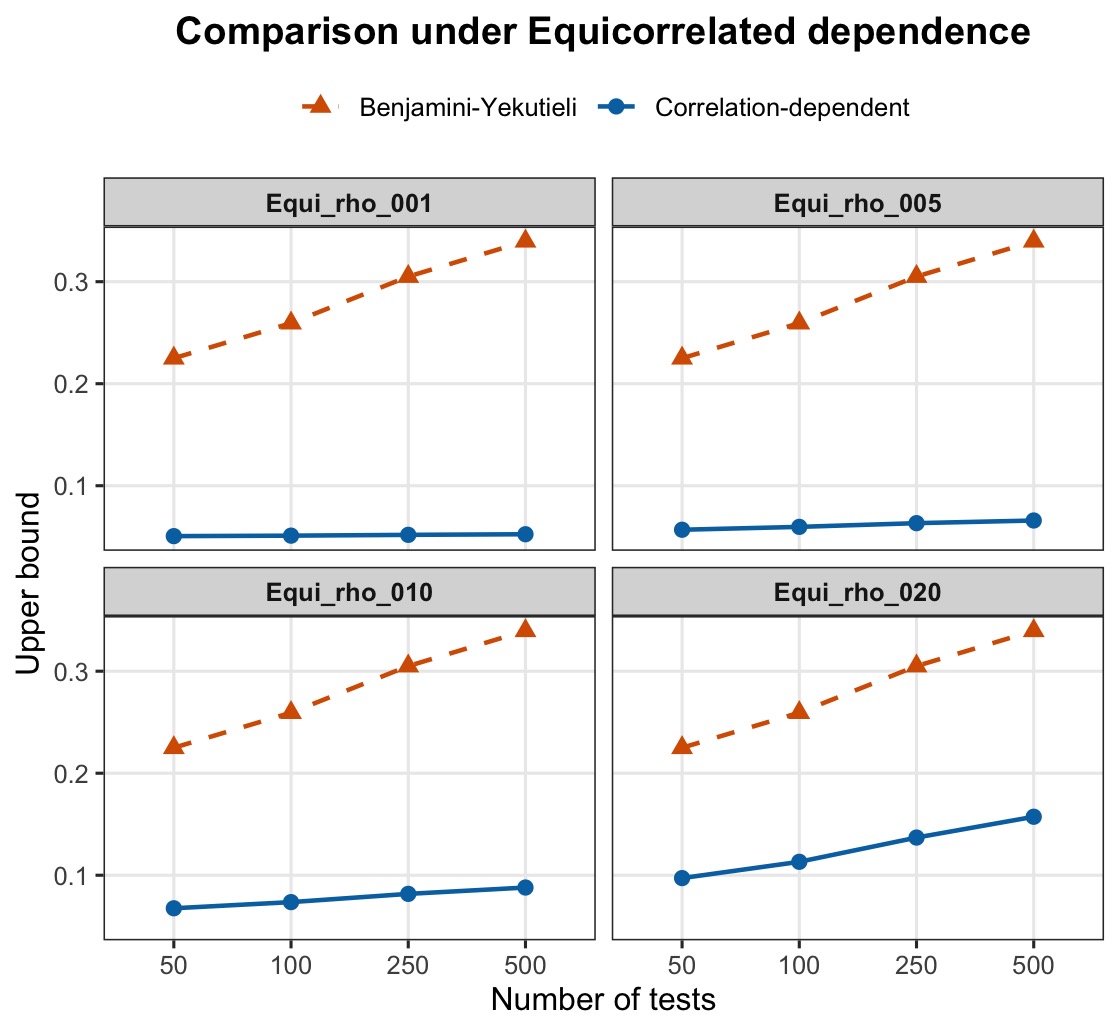}
        \caption{}
    \end{subfigure}
    \begin{subfigure}[t]{0.5\linewidth}
        \centering
        \includegraphics[width=\linewidth, height=0.27\textheight,
            keepaspectratio]{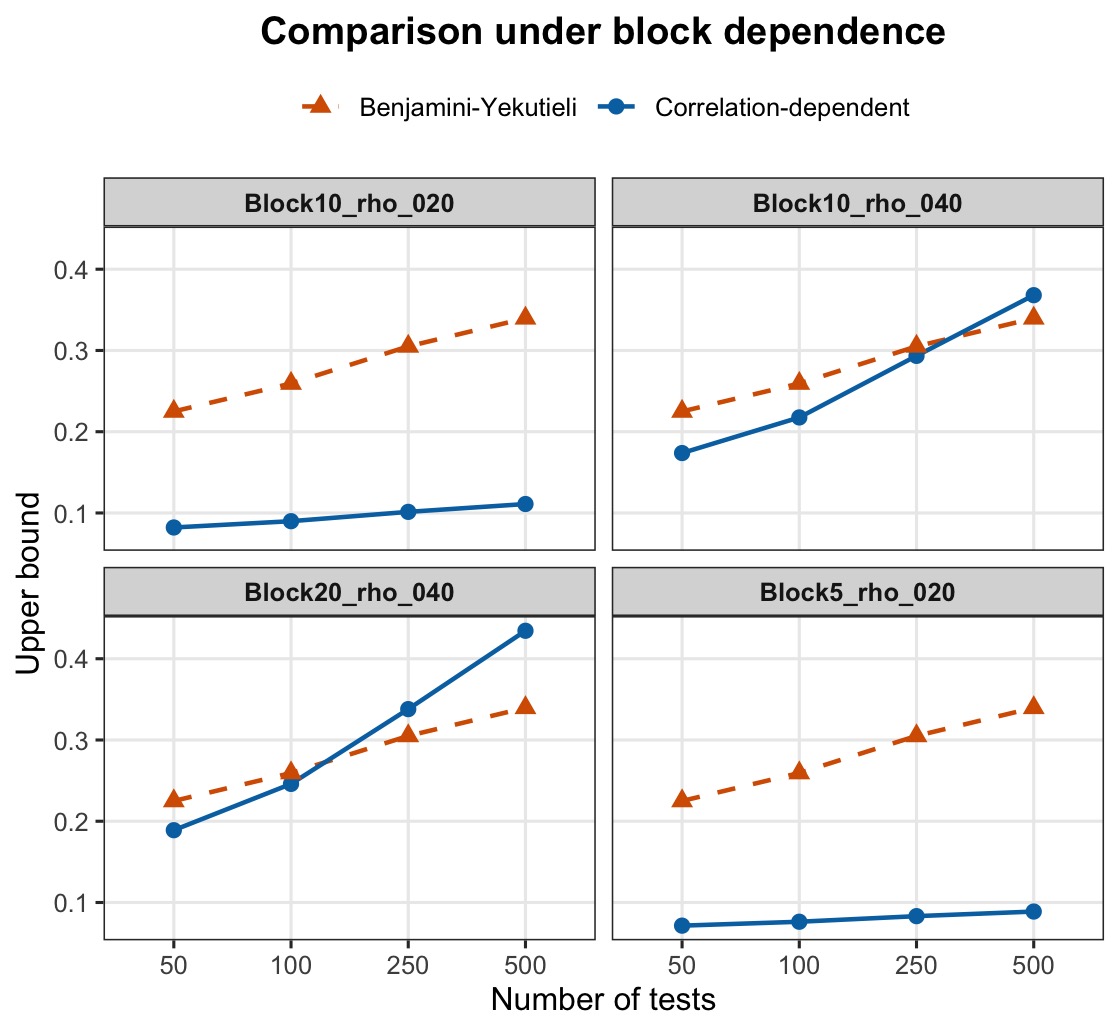}
        \caption{}
    \end{subfigure}\hfill

    \caption{
        The Benjamini-Yekutieli bound and the correlation dependent bound of the FDR displayed for different correlation structures $\bs{\Sigma}$ when $\alpha=0.05$
    }
    \label{fig:fdr-comparisons}
\end{figure}

Figure \ref{fig:fdr-comparisons} demonstrates some of the scenarios where the newer bound acts sharper than the BY bound. The sub-figure titles indicate the structure and the extent of correlations of $\bs{\Sigma}$. For example, \verb|AR1_rho_010| refers to an autoregressive of order $1$ matrix, that has $(i,j){th}$ element as $\rho^ {\mid i-j \mid}$, with correlations $\rho=0.10$. Similarly, \verb|Block20_rho_040| refers to a block diagonal matrix where each block of size 20 is equicorrelated with common correlation $\rho=0.40$. In the circumstances where all concerned elements share a common pairwise correlation $\rho$ among each other, the new bound behaves sharply for nominally small values of $\rho$, like $0.1, 0.2$, Fig~\ref{fig:fdr-comparisons}(b). The shifts $\tau_i$ can be explicit and equal in this case, that is $\tau_i = \tau = (1-\rho)\{1+(d-1)\rho\}/\{1+(d-2)\rho\}$. This provides a direct way of knowing if the correlation-dependent upper bound will be sharper for a given combination of $(d,\alpha)$ under a global null hypothesis. For the decaying correlation structure of Fig~\ref{fig:fdr-comparisons}(a), the new bound fails to be sharp at $\rho=0.5$. Indications of failure are also noticed in the block diagonal structure with increasing number of block size or increasing pairwise correlation. Not reducing the universality of the BY bound across arbitrary dependence and distributional assumptions, we proceed to inculcate a different perspective on the FDR bound of the BH method when correlations are known beforehand.

\subsection{Performance of CBSBH}
\label{subsec:CBSBH numerical}

Performance is evaluated using two criteria: empirical power, defined as the average proportion of false null hypotheses rejected, and empirical FDR, estimated as the Monte Carlo average of the false discovery proportion. Three types of correlation structures are studied -- equicorrelated, autoregressive of order 1 and inverse autoregressive of order 1 with the associated correlation parameter $\rho =0.5$. FDR gets controlled at $\alpha=0.05$. We vary the number of tests and, in a separate experiment, the Wishart degrees of freedom, where the signal strength varies in non-linear proportions of the degrees, to investigate how the amount of information available for covariance estimation affects performance. The confidence-bound shifted BH procedure is compared with BY, which provides finite-sample FDR control under arbitrary dependence, and with BH, which serves as a familiar reference procedure but does not possess a general finite-sample FDR guarantee under arbitrary dependence. The primary objective is to determine whether the proposed procedure can improve upon the power of BY while maintaining control of the FDR when a Gaussian model is appropriate and independent covariance information is available.

\begin{figure}[htbp!]
    \centering

    \begin{subfigure}[t]{0.33\linewidth}
        \centering
        \includegraphics[width=\linewidth, height=0.17\textheight]{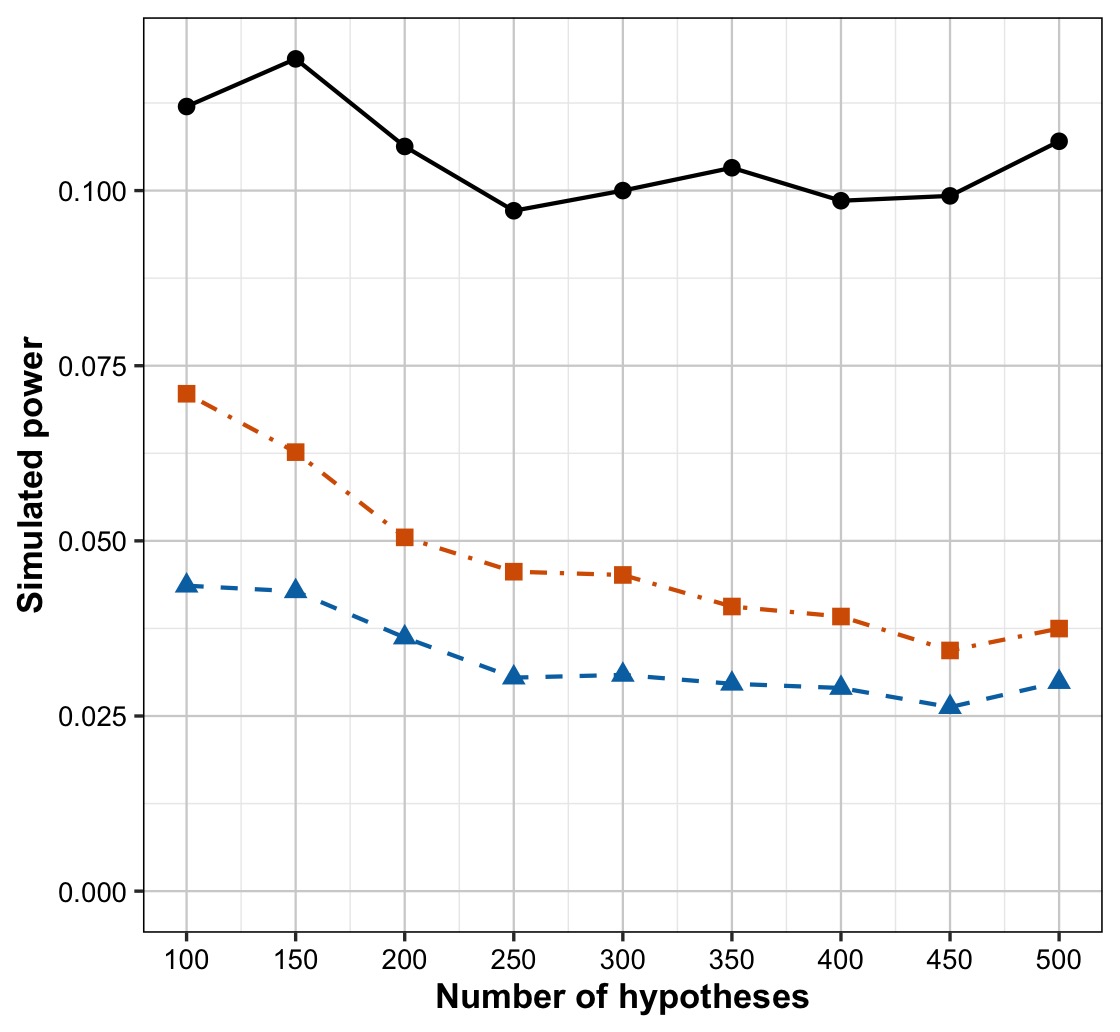}
        \caption{}
    \end{subfigure}\hfill
    \begin{subfigure}[t]{0.33\linewidth}
        \centering
        \includegraphics[width=\linewidth, height=0.17\textheight]{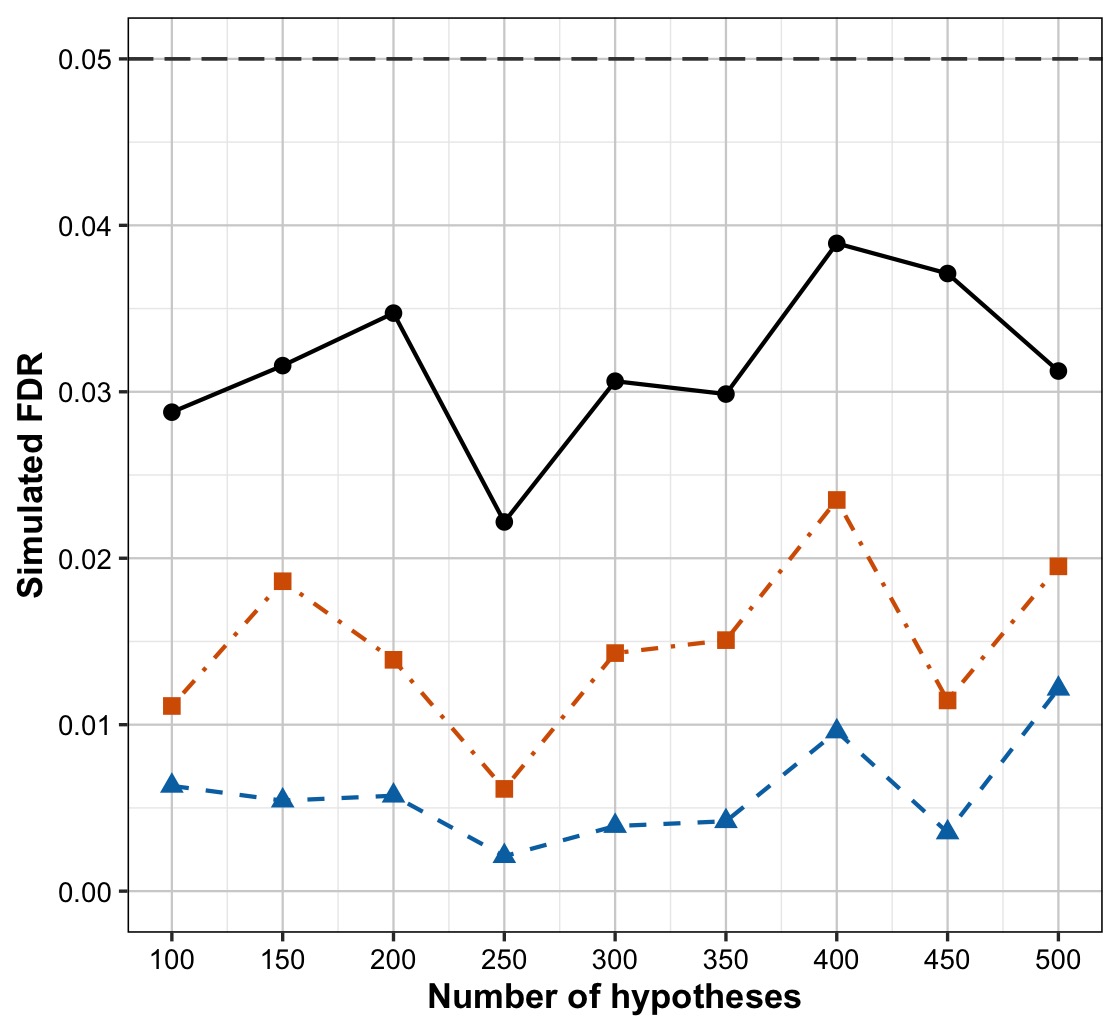}
        \caption{}
    \end{subfigure}\hfill
    \begin{subfigure}[t]{0.33\linewidth}
        \centering
        \includegraphics[width=\linewidth, height=0.17\textheight]{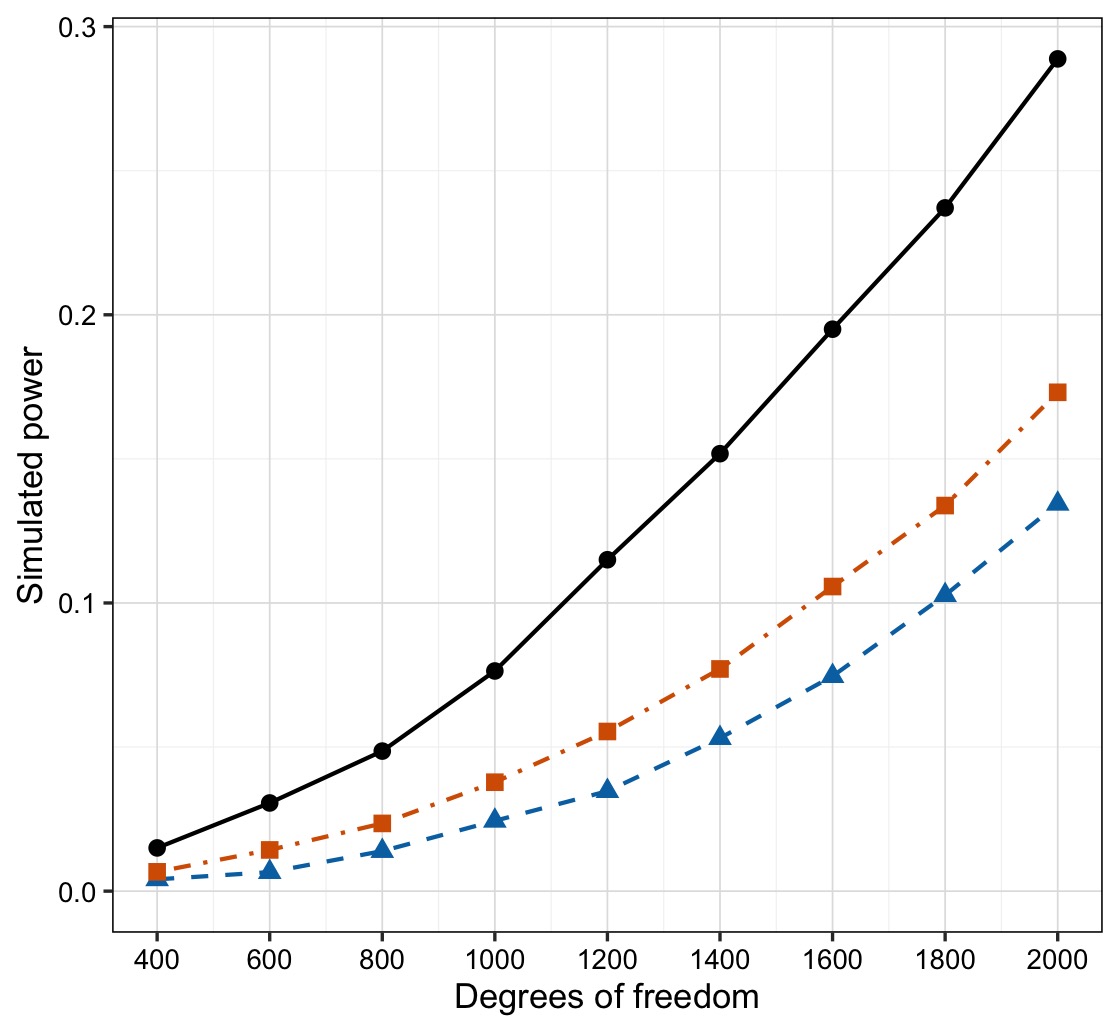}
        \caption{}
    \end{subfigure}\hfill

    \medskip

    \begin{subfigure}[t]{0.33\linewidth}
        \centering
        \includegraphics[width=\linewidth, height=0.17\textheight]{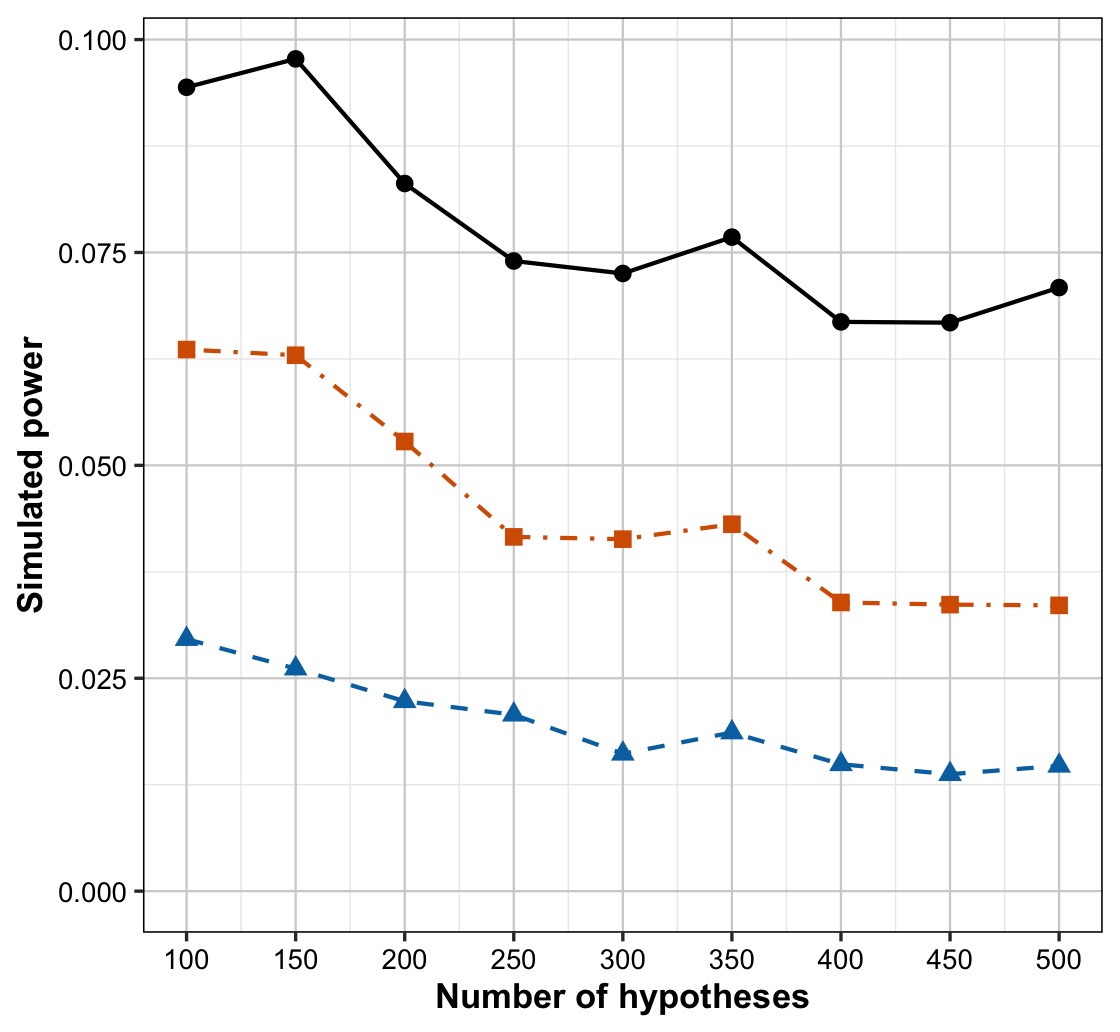}
        \caption{}
    \end{subfigure}\hfill
    \begin{subfigure}[t]{0.33\linewidth}
        \centering
        \includegraphics[width=\linewidth, height=0.17\textheight]{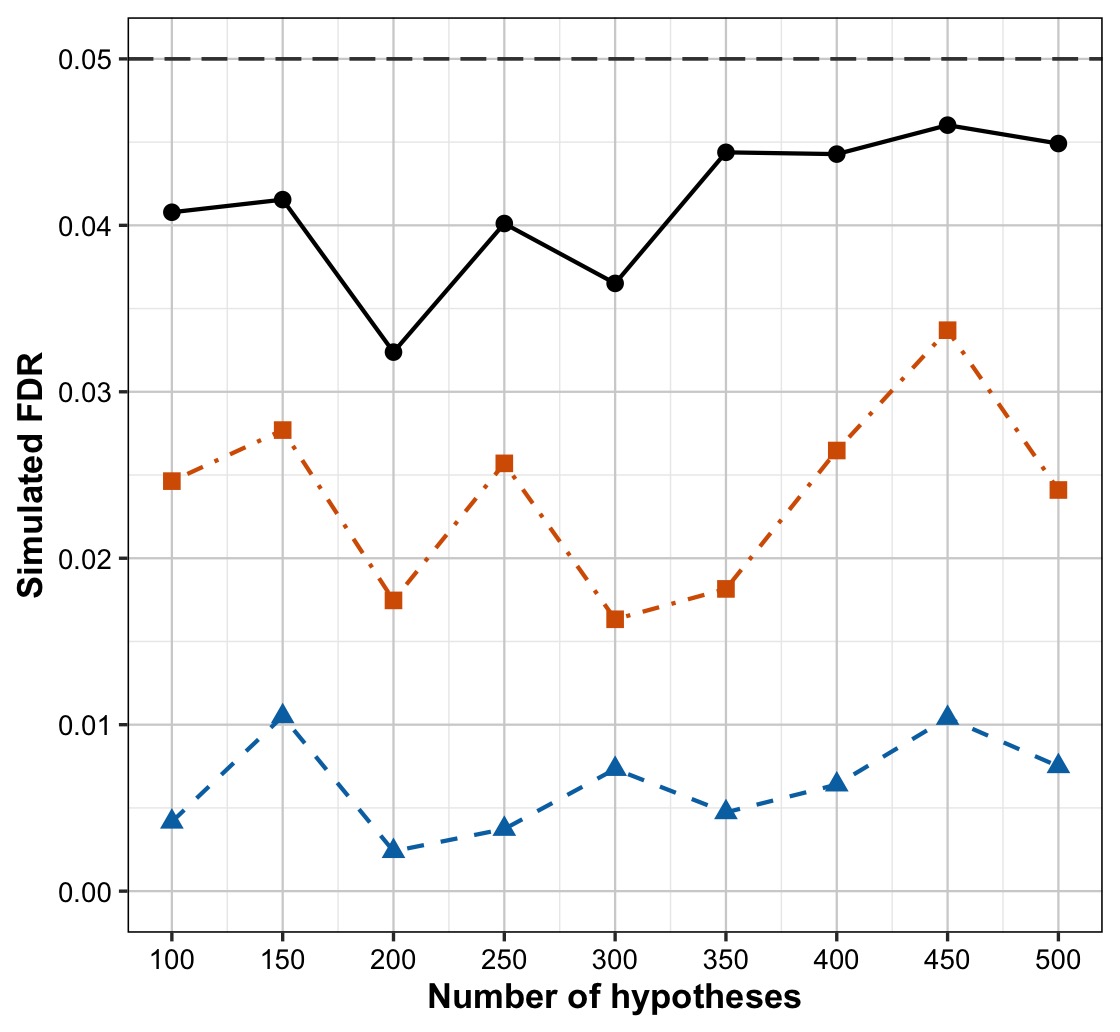}
        \caption{}
    \end{subfigure}\hfill
    \begin{subfigure}[t]{0.33\linewidth}
        \centering
        \includegraphics[width=\linewidth, height=0.17\textheight]{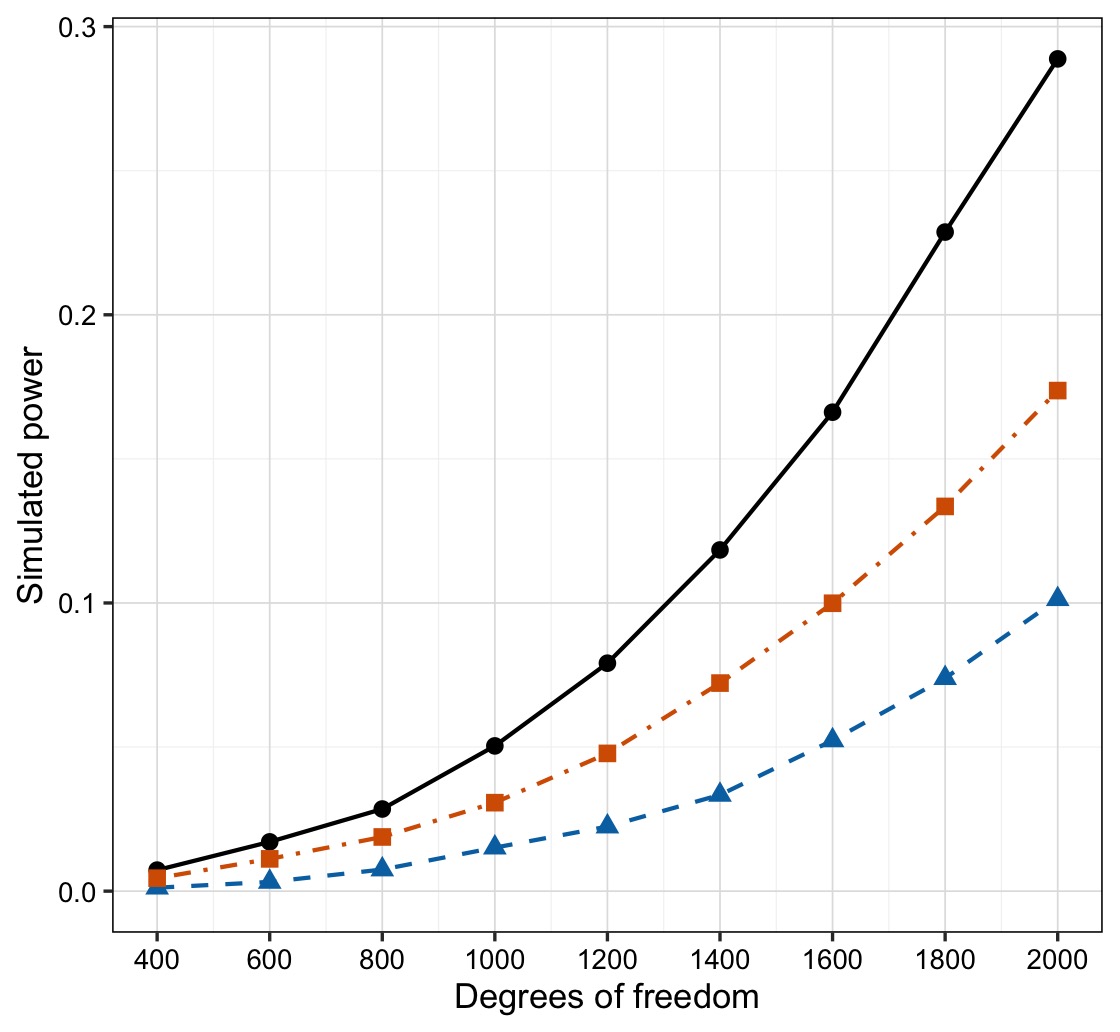}
        \caption{}
    \end{subfigure}\hfill

    \medskip

    \begin{subfigure}[t]{0.33\linewidth}
        \centering
        \includegraphics[width=\linewidth, height=0.17\textheight]{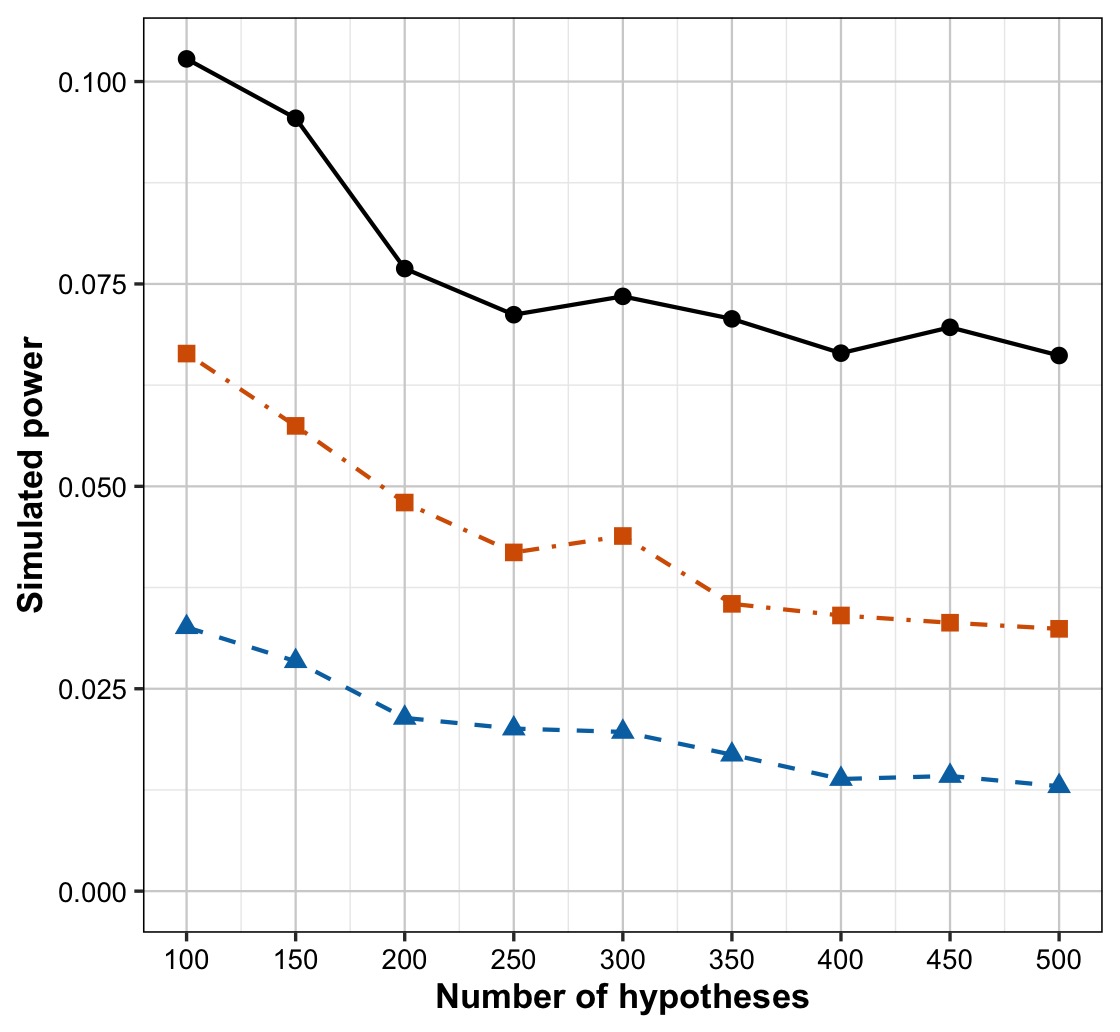}
        \caption{}
    \end{subfigure}\hfill
    \begin{subfigure}[t]{0.33\linewidth}
        \centering
        \includegraphics[width=\linewidth, height=0.17\textheight]{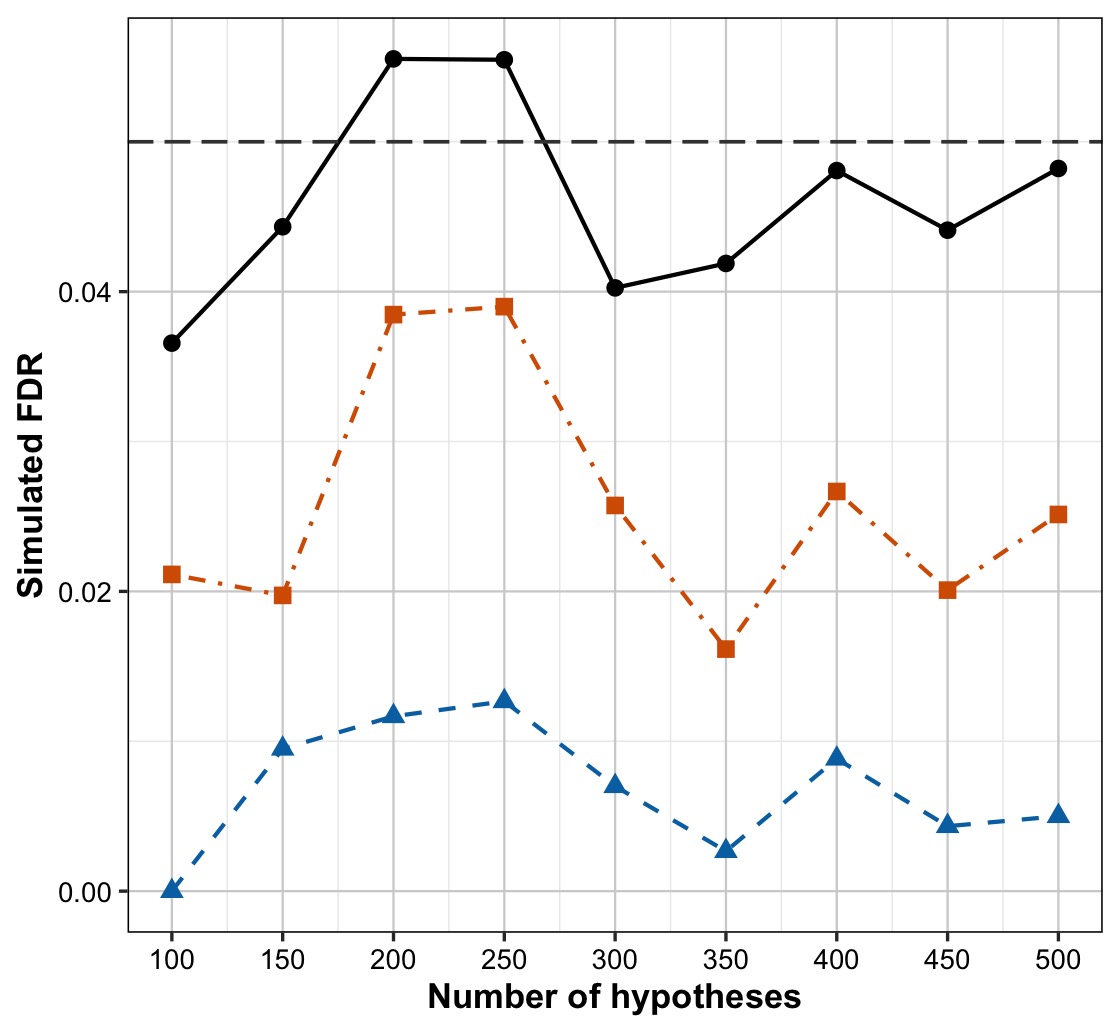}
        \caption{}
    \end{subfigure}\hfill
    \begin{subfigure}[t]{0.33\linewidth}
        \centering
        \includegraphics[width=\linewidth, height=0.17\textheight]{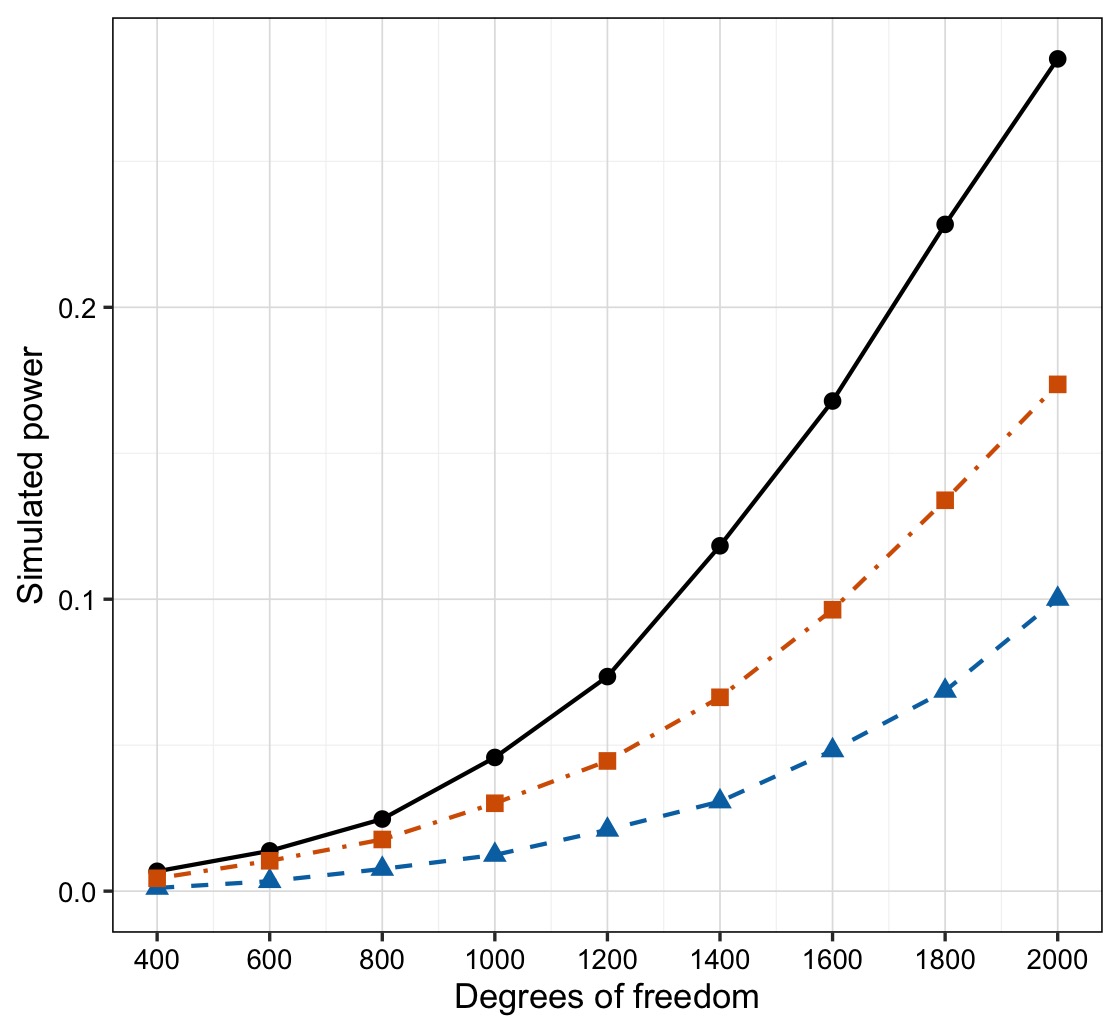}
        \caption{}
    \end{subfigure}\hfill
  \caption{
        Line plots are displayed for a compound symmetric (first row), autoregressive (second row) and inverse autoregressive (third row) correlation structures with $\rho=0.5$ and confidence level $\beta=0.01$. Methods compared are Benjamini-Hochberg (Circle and black), Benjamini-Yekutieli (Triangle point up and blue) and confidence-bound shifted BH (Square and orange). Simulated power is observed across varying number of tests (left panel) and degrees of freedom (right panel) while simulated FDR is plotted across varying number of tests (middle panel)
        }
  \label{fig:means-unknown-1} 
\end{figure}

\begin{figure}[htbp!]
    \centering

    \begin{subfigure}[t]{0.33\linewidth}
        \centering
        \includegraphics[width=\linewidth, height=0.2\textheight]{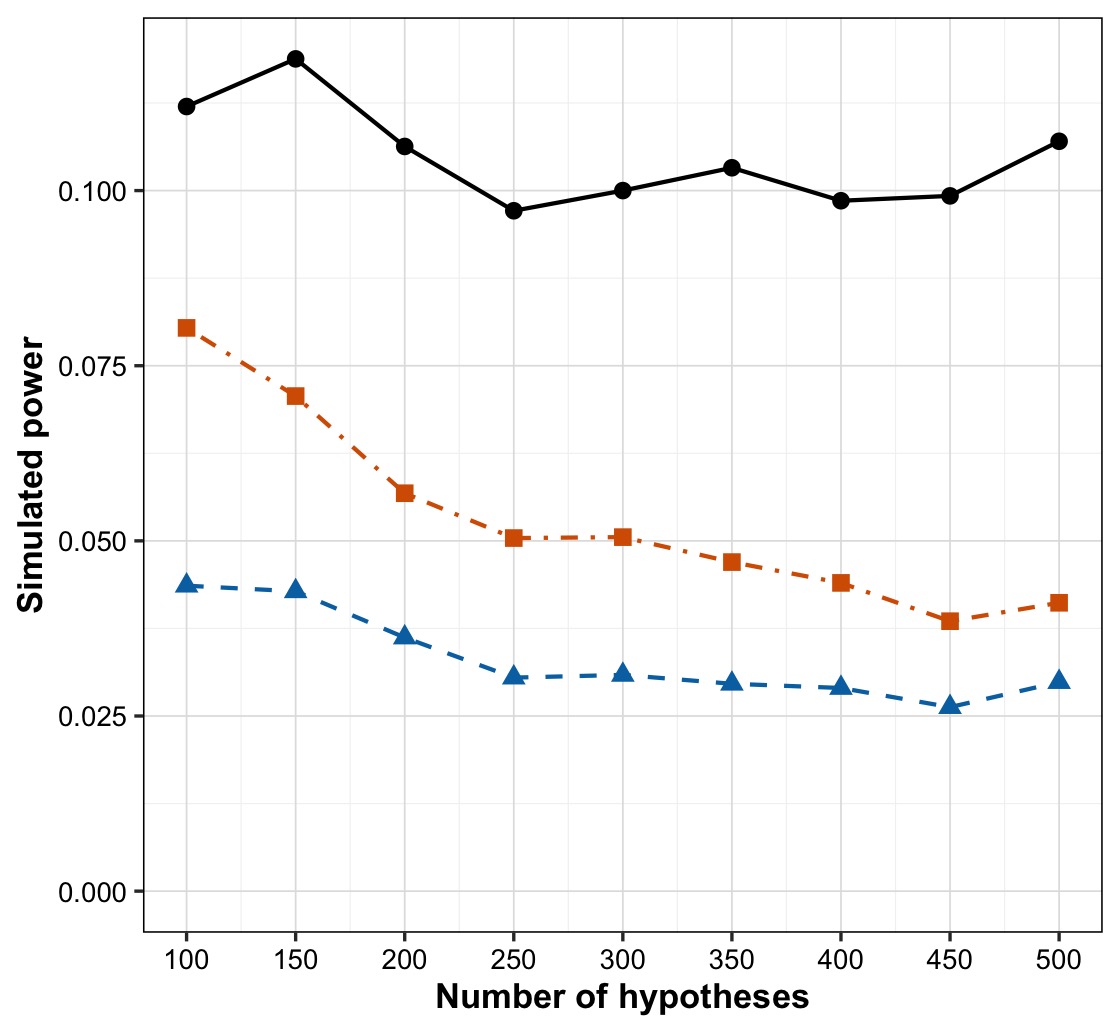}
        \caption{}
    \end{subfigure}\hfill
    \begin{subfigure}[t]{0.33\linewidth}
        \centering
        \includegraphics[width=\linewidth, height=0.2\textheight]{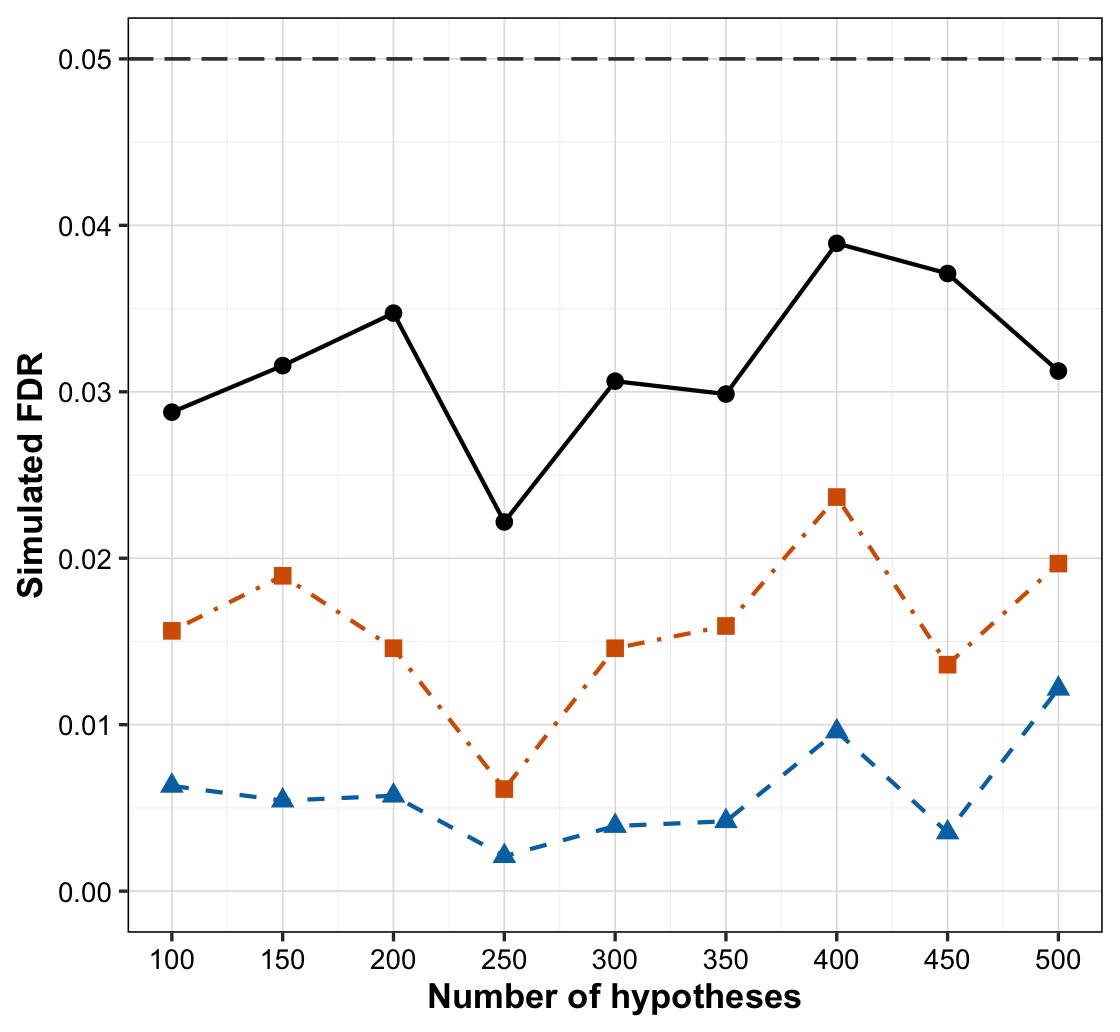}
        \caption{}
    \end{subfigure}\hfill
    \begin{subfigure}[t]{0.33\linewidth}
        \centering
        \includegraphics[width=\linewidth, height=0.2\textheight]{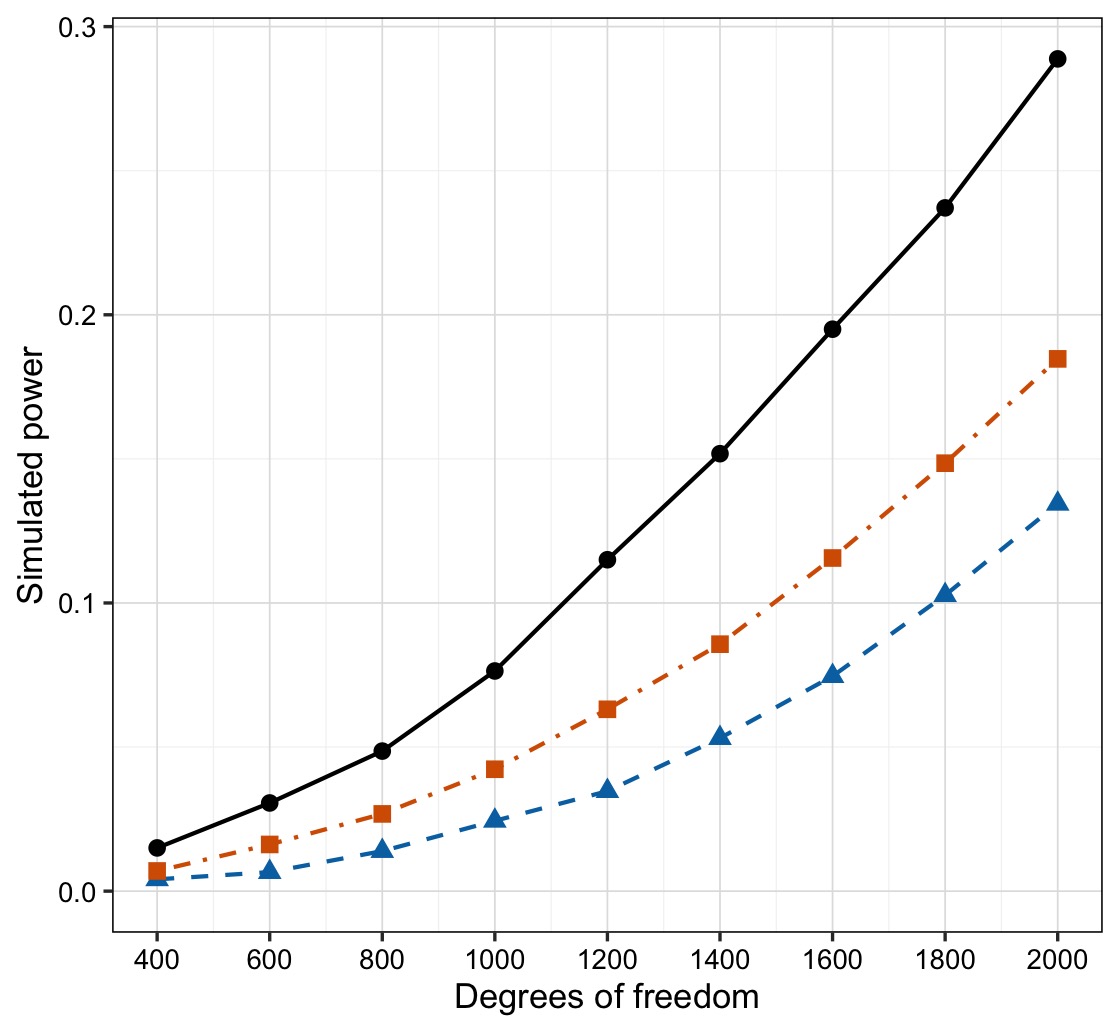}
        \caption{}
    \end{subfigure}\hfill

    \medskip

    \begin{subfigure}[t]{0.33\linewidth}
        \centering
        \includegraphics[width=\linewidth, height=0.2\textheight]{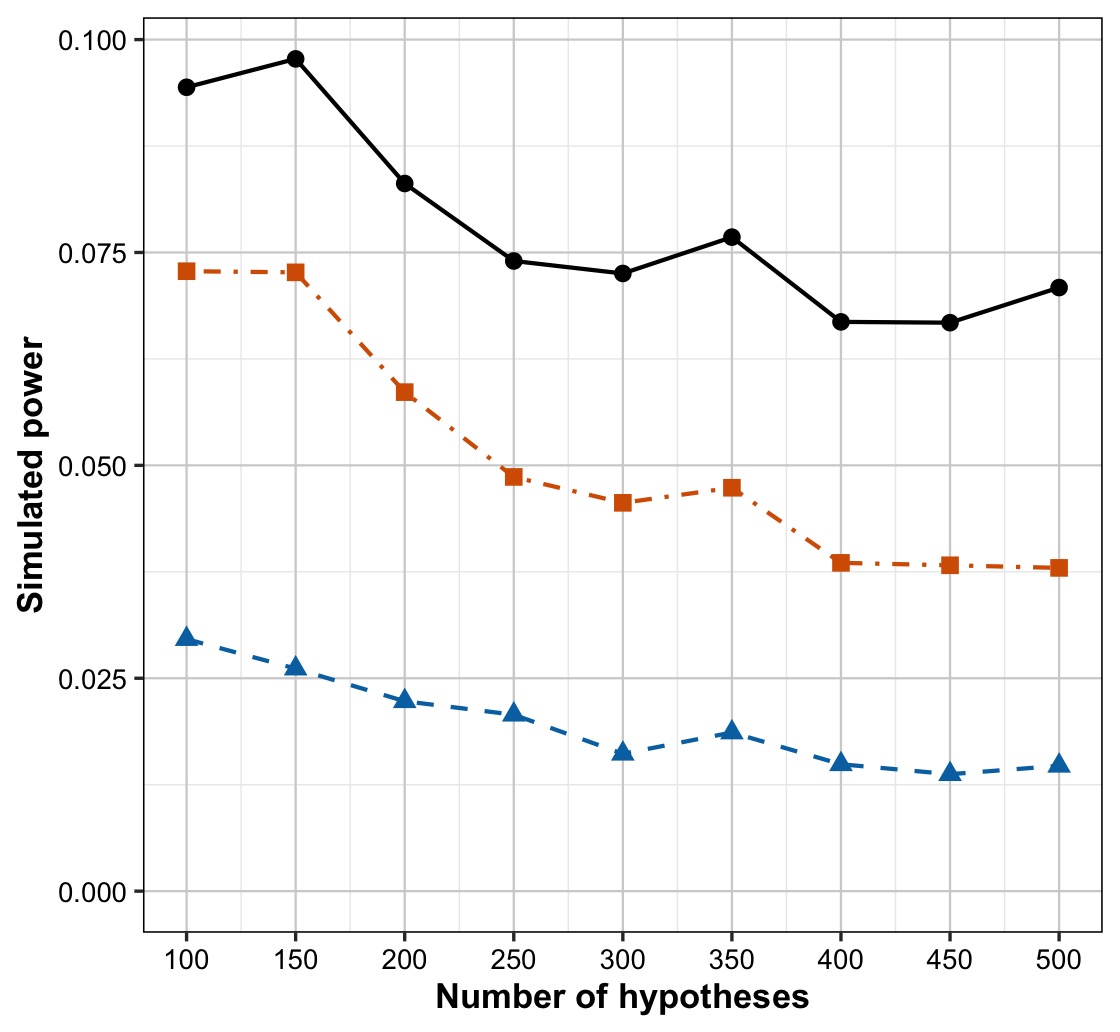}
        \caption{}
    \end{subfigure}\hfill
    \begin{subfigure}[t]{0.33\linewidth}
        \centering
        \includegraphics[width=\linewidth, height=0.2\textheight]{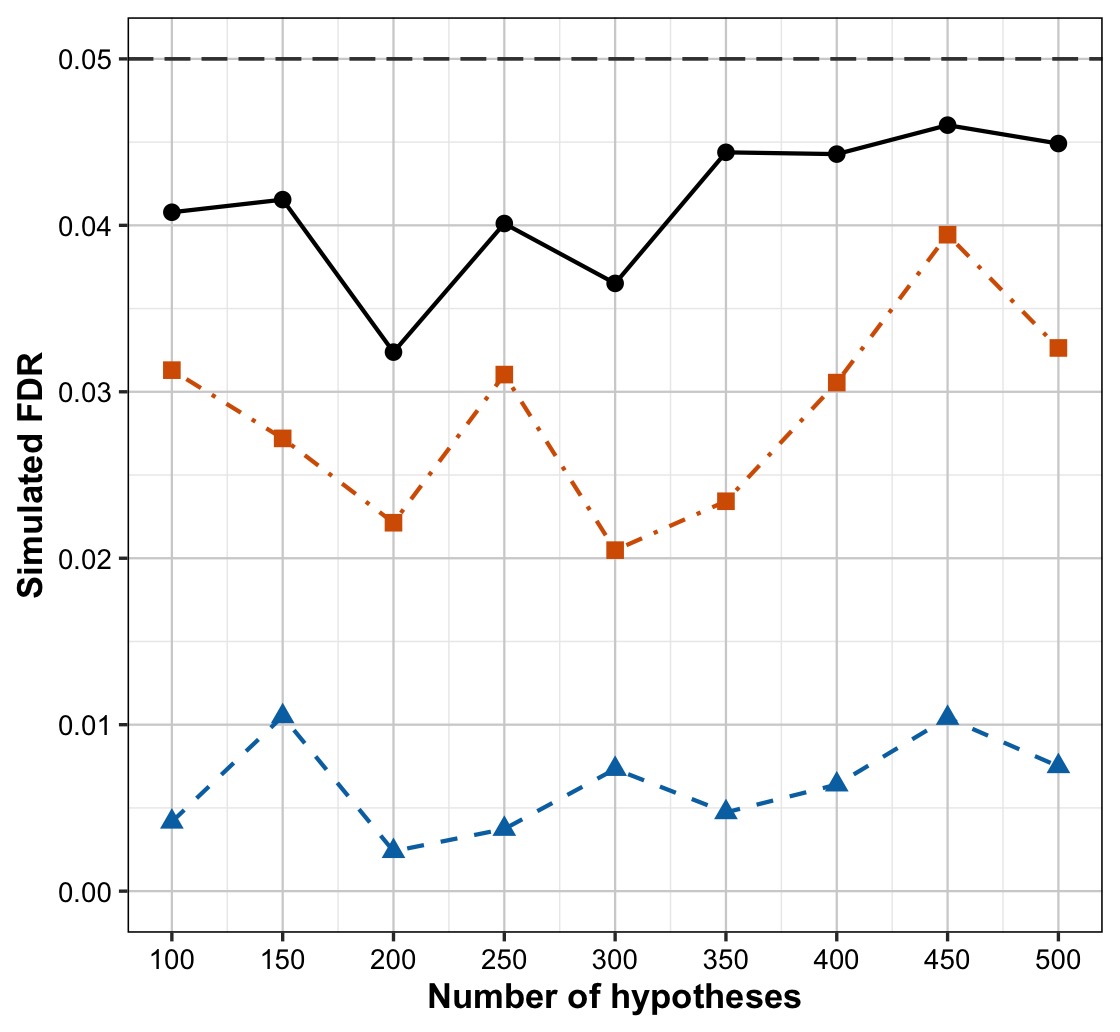}
        \caption{}
    \end{subfigure}\hfill
    \begin{subfigure}[t]{0.33\linewidth}
        \centering
        \includegraphics[width=\linewidth, height=0.2\textheight]{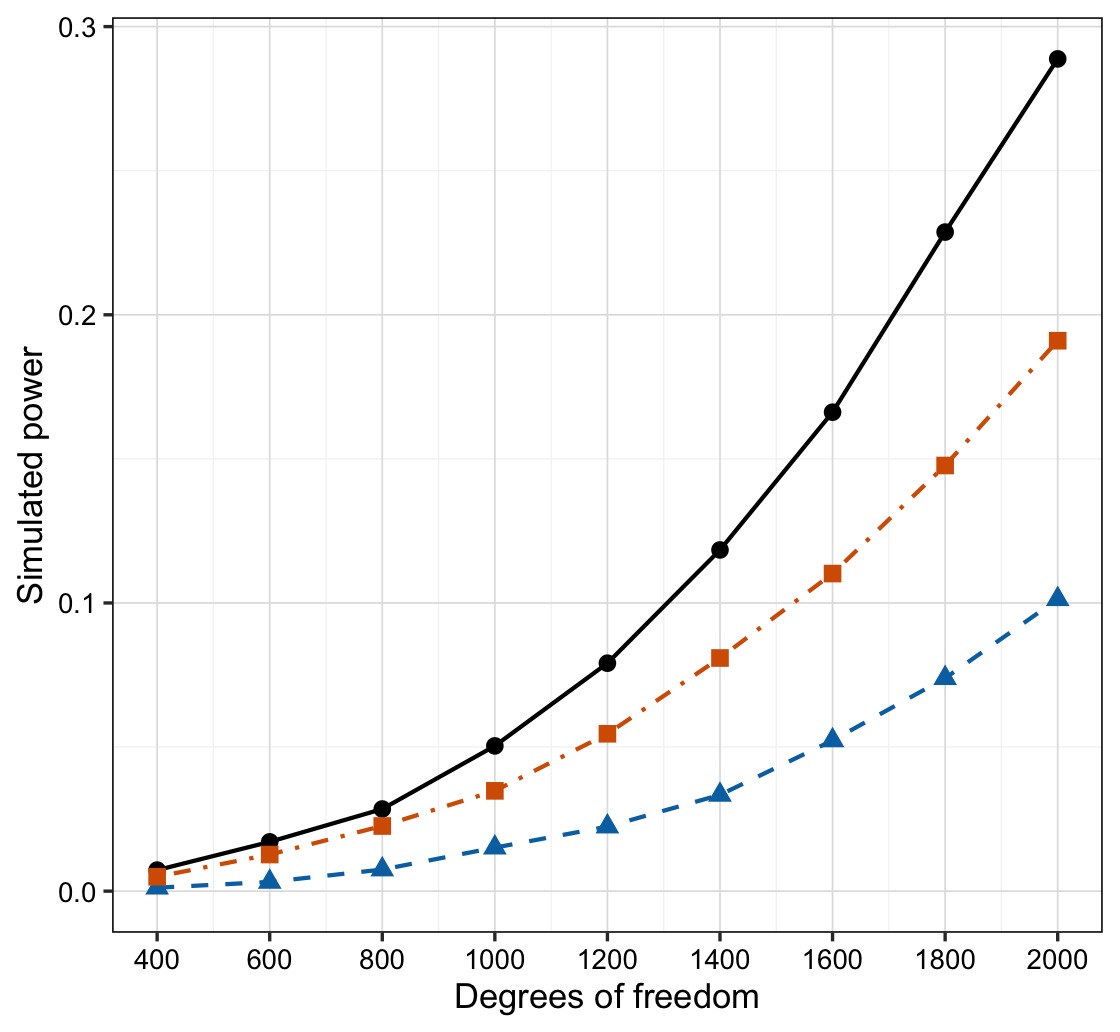}
        \caption{}
    \end{subfigure}\hfill

    \medskip

    \begin{subfigure}[t]{0.33\linewidth}
        \centering
        \includegraphics[width=\linewidth, height=0.2\textheight]{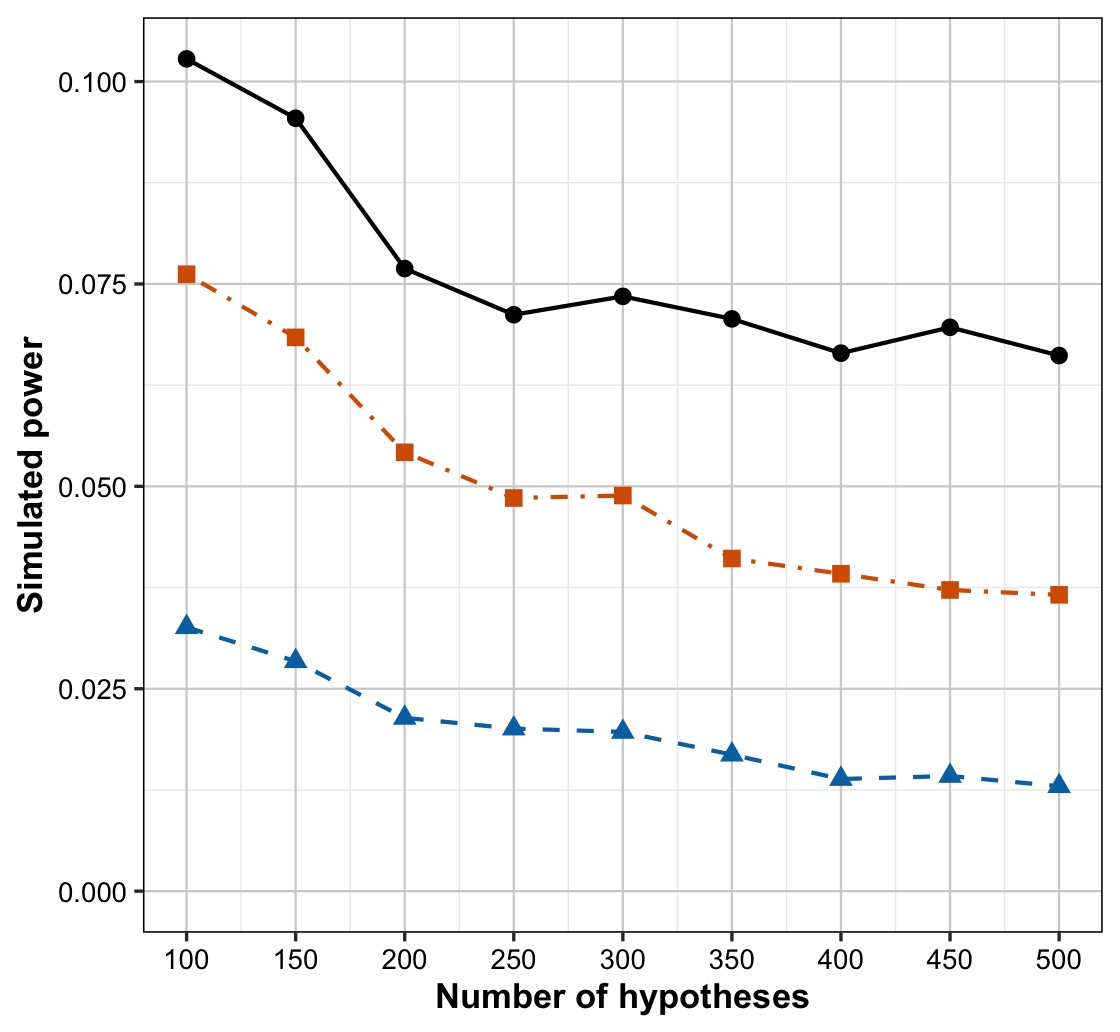}
        \caption{}
    \end{subfigure}\hfill
    \begin{subfigure}[t]{0.33\linewidth}
        \centering
        \includegraphics[width=\linewidth, height=0.2\textheight]{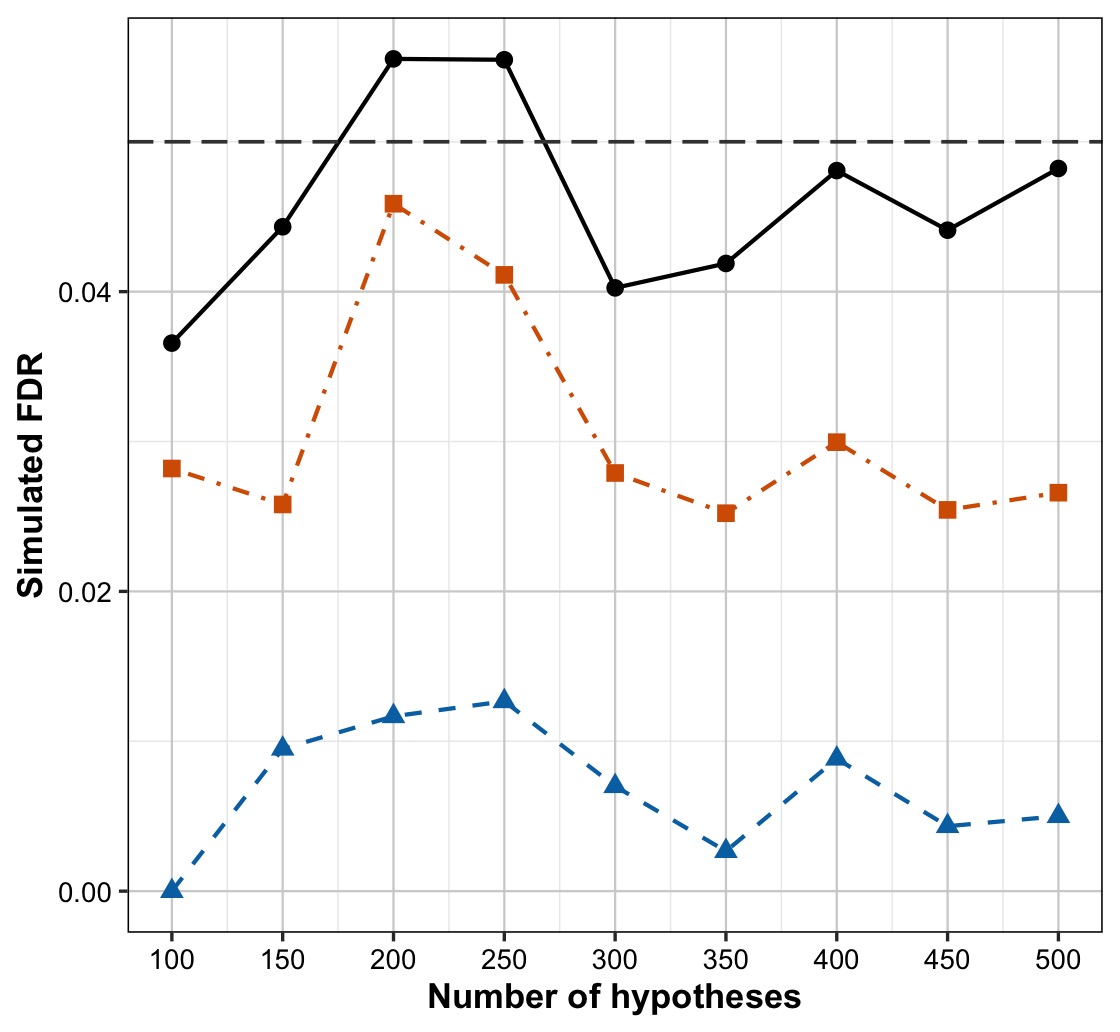}
        \caption{}
    \end{subfigure}\hfill
    \begin{subfigure}[t]{0.33\linewidth}
        \centering
        \includegraphics[width=\linewidth, height=0.2\textheight]{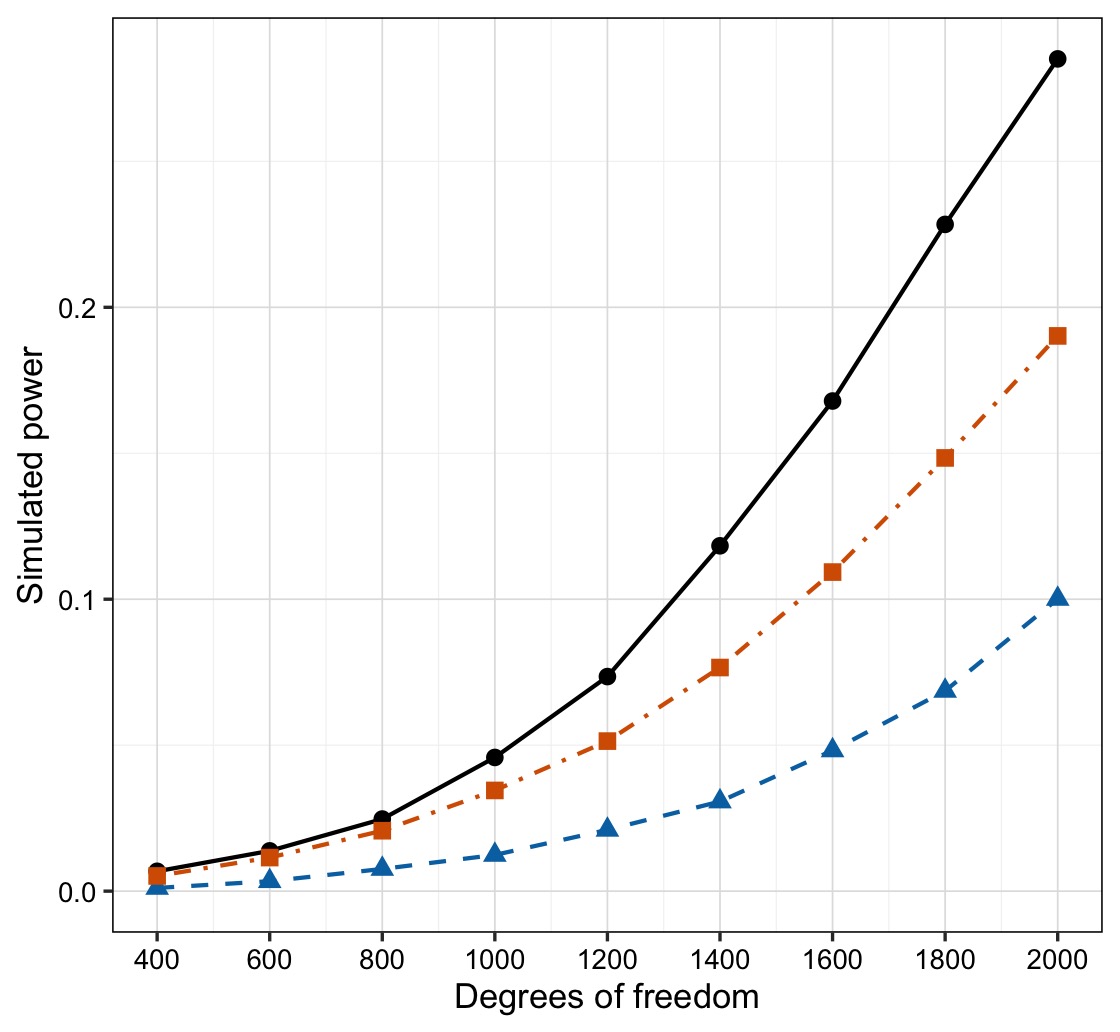}
        \caption{}
    \end{subfigure}\hfill
  \caption{
        Line plots are displayed for a compound symmetric (first row), autoregressive (second row) and inverse autoregressive (third row) correlation structures with $\rho=0.5$ and confidence level $\beta=0.0001$. Methods compared are Benjamini-Hochberg (Circle and black), Benjamini-Yekutieli (Triangle point up and blue) and confidence-bound shifted BH (Square and orange). Simulated power is observed across varying number of tests (left panel) and degrees of freedom (right panel) while simulated FDR is plotted across varying number of tests (middle panel)
        }
  \label{fig:means-unknown-2} 
\end{figure}
For constructing the simultaneous lower confidence bounds, the componentwise miscoverage probabilities $\beta_i, \; i=1, \ldots, d$, are chosen to satisfy $\sum_{i=1}^{d}\beta_i\leq \beta$, with $\beta$  being the overall miscoverage probability. We set $\beta \in \{0.01,0.0001\}$, corresponding to simultaneous confidence levels of $0.99$ and $0.9999$, respectively, and use the equal Bonferroni allocation $\beta_i= \beta/d,\; i=1,\ldots,d$. 

Across the correlation structures considered, the confidence-bound shifted BH procedure identifies non-null signals more effectively than BY. Figures~\ref{fig:means-unknown-1} and \ref{fig:means-unknown-2} show that this power advantage persists as both the number of tests and the degrees of freedom available for covariance estimation are varied. These results support the proposed procedure as a competitive alternative to BY when the Gaussian assumption is reasonable and an independent covariance estimate is available.

Among the two choices of the overall miscoverage probability, $\beta=0.0001$ generally yields higher empirical power than $\beta=0.01$, while also producing tighter empirical control of the FDR. Thus, within the simulation configurations considered, the more stringent simultaneous confidence level improves both signal-detection performance and the stability of FDR control.

\section{Discussion}\label{sec:discussion}
This paper extends the PTDN and generalized shifted-BH framework of \citet{GhoshSarkar2025}, broadening both its interpretation and scope. Revisiting PTDN clarifies its role as a notion of positive dependence tailored to the lower-tail events underlying BH-type procedures and its relationship with the classical PRDS condition. This perspective suggests that PTDN is useful not only for constructing finite-sample FDR-controlling procedures, but also for identifying the aspects of dependence most relevant to the BH FDR.

This viewpoint yields a dependence-adaptive analysis of the original BH procedure. The conditional variance parameters, which underlie shifted-BH methods, also provide explicit finite-sample bounds on the BH FDR. Unlike the generic bounds of \citet{BenjaminiYekutieli2001} and \citet{Su2018}, these bounds incorporate the specified covariance structure and recover the exact BH FDR under independence. Their adaptation is particularly transparent under equal correlations, where they vary explicitly with the correlation parameter. They therefore complement the common-factor asymptotic theory of \citet{Lei2026}, which addresses worst-case behavior over broader dependence classes. The same conditional-variance structure also explains how coordinate-specific shifting can create a rejection advantage over BH when signals occur in highly dependent coordinates.

The unknown-covariance development further extends shifted-BH methodology. Replacing the oracle conditional variance parameters by simultaneous lower confidence bounds obtained from an independent Wishart sample yields a confidence-bound shifted BH procedure with finite-sample FDR control. More generally, this construction suggests that uncertainty in the dependence parameters required for PTDN can be incorporated into the testing rule by allocating the target error level between FDR calibration and confidence failure.

Several questions remain. For equi-correlated and common-factor Gaussian models, it would be interesting to determine whether the covariance-adaptive BH bound recovers the optimal asymptotic order established by \citet{Lei2026}, while providing sharper finite-sample information over practically relevant dependence regimes. For unknown covariance structures, sharper simultaneous confidence bounds, adaptive allocation of the confidence error, and extensions to regularized high-dimensional covariance estimation may improve the methodology. Other promising directions include data-fission or external-randomization constructions and identifying broader distributional families possessing the tail-concavity properties underlying PTDN.

Taken together, these results broaden PTDN from a principle for constructing shifted-BH procedures into a more general framework for understanding the finite-sample behavior of BH, identifying when dependence-aware shifting can improve discovery, and extending FDR control to settings in which the dependence structure itself must be estimated. Although the present developments are established under Gaussianity, the Gaussian model is used not as a universal description of modern dependent data, but as a tractable setting that provides exact finite-sample distributions, useful approximations for standardized mean contrasts, score statistics, and other aggregated estimators, and a transparent basis for analyzing dependence-adjusted procedures. The resulting principles, particularly those formulated through conditional lower-tail behavior, may therefore provide a foundation for extending PTDN to broader classes of dependent testing problems beyond the Gaussian setting.

\section*{Acknowledgements}
The authors gratefully acknowledge support from NSF grant DMS 2210687. The authors also credit OpenAI GPT-5.6 Sol for assisting during the preparation of this work in terms of language editing and improving the clarity and readability of certain parts of the manuscript. After using this tool, the authors reviewed the content and take full responsibility for the content of the publication.

\bibliographystyle{apalike}
\bibliography{biblio.bib}

\end{document}